\documentclass[12pt, draftclsnofoot, onecolumn]{IEEEtran}
\usepackage{graphicx}
\usepackage{amsthm}
\usepackage{epsfig}
\usepackage{latexsym}
\usepackage{amsfonts}
\usepackage{here}
\usepackage{rawfonts}
\usepackage[latin1]{inputenc}
\usepackage[T1]{fontenc}
\usepackage{calc}
\usepackage{capitalgreekitalic}
\usepackage{url}
\usepackage{enumerate}
\usepackage{color}
\usepackage[tbtags]{amsmath}
\usepackage{amssymb}
\usepackage{upref}
\usepackage{epic,eepic}
\usepackage{times}
\usepackage{dsfont}
\usepackage{comment}
\usepackage{cite}
\usepackage{epstopdf}

\renewcommand{\Pr}{\ensuremath{\operatorname{Pr}}}

\newtheorem{theorem}{\bf Theorem}
\newtheorem{proposition}{\bf Proposition}

\usepackage{dsfont}

\newcounter{step}
\newlength{\totlinewidth}
  {\end{list}%
  \rule{\linewidth}{1pt}}
\newcounter{substep}

  {\end{list}}

\newlength{\aligntop}
\newlength{\alignbot}
 \makeatletter
\renewenvironment{align}{%
  \vspace{\aligntop}
  \start@align\@ne\st@rredfalse\m@ne
}{%
  \math@cr \black@\totwidth@
  \egroup
  \ifingather@
    \restorealignstate@
    \egroup
    \nonumber
    \ifnum0=`{\fi\iffalse}\fi
  \else
    $$%
  \fi
  \ignorespacesafterend%
  \vspace{\alignbot}\par\noindent
} \makeatother

\IEEEoverridecommandlockouts

\usepackage{algorithm}
\usepackage{algorithmic}
\usepackage{subfigure}

\makeatletter
\newcommand\semihuge{\@setfontsize\semihuge{18.2}{22}}
\makeatother

\makeatletter
\newcommand\semismall{\@setfontsize\semihuge{12.4}{15}}
\makeatother

\begin{document}
\clearpage
\title{\semihuge Caching in the Sky: Proactive Deployment of Cache-Enabled Unmanned Aerial Vehicles for Optimized Quality-of-Experience}
  
  \author{\IEEEauthorblockN{\semismall Mingzhe Chen$^1$, Mohammad Mozaffari$^2$, Walid Saad$^{2,4}$, Changchuan Yin$^1$,\\ M\'erouane Debbah$^3$, and Choong-Seon Hong$^4$}\vspace{0.05cm}\\

	\IEEEauthorblockA{
		\small $^1$ Beijing Laboratory of Advanced Information Network, Beijing University of Posts and Telecommunications, Beijing, China 100876, Emails: \url{chenmingzhe@bupt.edu.cn} and \url{ccyin@ieee.org.}\\
		$^2$ Wireless@VT, Electrical and Computer Engineering Department, Virginia Tech, VA, USA,\\ Emails:\url{{mmozaff , walids}@vt.edu}.\\
		$^3$ Mathematical and Algorithmic Sciences Lab, Huawei France R \& D, Paris, France, \\Email: \url{merouane.debbah@huawei.com}.\\
		$^4$ Department of Computer Science and Engineering, Kyung Hee University, Yongin, South Korea, \\Email: \protect\url{cshong@khu.ac.kr.}\\ 
	}\vspace{-0.9cm}}


%
%
%
%

\maketitle
\thispagestyle{empty}
\vspace{0cm}
\begin{abstract}
\vspace{-0.1cm}
In this paper, the problem of proactive deployment of cache-enabled unmanned aerial vehicles (UAVs) for optimizing the quality-of-experience (QoE) of wireless devices in a cloud radio access network (CRAN) is studied. In the considered model, the network can leverage human-centric information such as users' visited locations, requested contents, gender, job, and device type to predict the content request distribution and mobility pattern of each user. Then, given these behavior predictions, the proposed approach seeks to find the user-UAV associations, the optimal UAVs' locations, and the contents to cache at UAVs. This problem is formulated as an optimization problem whose goal is to maximize the users' QoE while minimizing the transmit power used by the UAVs. To solve this problem, a novel algorithm based on the machine learning framework of conceptor-based echo state networks (ESNs) is proposed. Using ESNs, the network can effectively predict each user's content request distribution and its mobility pattern when limited information on the states of users and the network is available. 
Based on the predictions of the users' content request distribution and their mobility patterns, we derive the optimal user-UAV association,  optimal locations of the UAVs as well as the content to cache at UAVs. Simulation results using real pedestrian mobility patterns from BUPT and actual content transmission data from Youku show that the proposed algorithm can yield 40\% and 61\%
gains, respectively, in terms of the average transmit power and the percentage of the users with satisfied QoE compared to a benchmark algorithm without caching and a benchmark solution without UAVs. 

\end{abstract}\vspace{-0.5cm}

\vspace{0cm}
\vspace{-0cm}

\section{Introduction}
The next-generation of cellular systems is expected to be largely user centric and, as such, it must be cognizant of human-related information such as users' behavior, mobility patterns, and quality-of-experience (QoE) expectations \cite{Chen2016User}. 
One promising approach to introduce such wireless network designs with human-in-the-loop is through the use of cloud radio access networks (CRANs) \cite{MugenRecent}. In CRANs, a central cloud processor can parse through the massive users' data to learn the users' information such as content request distribution and mobility patterns and, then, determine how to manage resources in the network. 
However, an effective exploration of human-in-the-loop features in a CRAN faces many challenges that range from effective predictions to user behavior tracking, effective caching, and optimized resource management.


Some recent works such as in\cite{Nguyen2012Extracting,Lee2006Modeling,UserinitiatedData,Song2010Limits,Bigdata,SoysaPredicting, NagarajaCaching,EchoStateNetworks2,ProactiveCaching2016,Semiari2015Context,Zhang2015Social} have studied a number of ideas related to CRANs with human-in-the-loop. 
In \cite{Nguyen2012Extracting} and \cite{Lee2006Modeling}, a new approach is proposed for predicting the users' mobility patterns using a deep learning algorithm and a semi-Markov process. The authors in \cite{UserinitiatedData} proposes a new type of user-initiated network for cellular users to trade data plans by leveraging personal hotspots with users' smartphones. The work in \cite{Song2010Limits} investigated the probability of predicting users' mobility patterns. Nevertheless, the mobility prediction works in\cite{Nguyen2012Extracting,Lee2006Modeling,Song2010Limits,UserinitiatedData} focused only on the prediction phase and did not study how the users' mobility patterns can be used to optimize the wireless performance using user-centric caching and resource allocation techniques.   The authors in \cite{Bigdata} developed a data extraction method using the Hadoop platform to predict content popularity. The work in \cite{SoysaPredicting} proposed a fast threshold spread model to predict the future access patterns of multimedia content based on social information. In \cite{NagarajaCaching}, the authors exploited the instantaneous demands of wireless users to estimate the content popularity and devise an optimal random caching strategy. The authors in \cite{EchoStateNetworks2} proposed an echo state network to predict the users' content request distribution and mobility patterns in CRANs. The work in \cite{ProactiveCaching2016} proposed a caching-based milimiter wave (mmWave) framework, which pre-caches video contents at the base stations for handover users to reduce the connection delay. In \cite{Semiari2015Context} and \cite{Zhang2015Social}, the authors proposed a novel resource allocation approach based on the social context of wireless users.  
However, most of these existing caching works\cite{Bigdata,SoysaPredicting,NagarajaCaching,EchoStateNetworks2,ProactiveCaching2016,Semiari2015Context,Zhang2015Social} were typically restricted to static networks without mobility and ultra dense users.
Note that, in these contributions \cite{Bigdata,SoysaPredicting,NagarajaCaching,EchoStateNetworks2,ProactiveCaching2016,Semiari2015Context,Zhang2015Social}, the cache content is stored at the terrestrial static base stations. However, the cache content at a static base station cannot be effectively used to serve the mobile users as they move outside the coverage range of the base station. In addition, while a mobile user moves to a new cell, its requested content may not be available at the new base station and, consequently, the users cannot be serviced properly. In such case, to serve the mobile user, one content needs to cached at multiple base stations which will not be efficient. Therefore, there is a need to track the users' mobility patterns so as to improve the caching efficiency. To this end, unmanned aerial vehicles (UAVs) can be used as flying base stations to dynamically cache the popular contents, track the mobility pattern of the corresponding users and, then, effectively serve them. In this case, due to the high altitude  and flexible deployment of the UAVs, they can establish reliable communication links to the users by mitigating the blockage effect.

The use of UAVs for enhancing wireless communications in cellular and ad hoc networks was studied in \cite{Mozaffari2016Unmanned,ThroughputMaximization, Letter, Modelingairtoground,bor,kalantari,MozaffariIoT}.
However, this existing literature \cite{Mozaffari2016Unmanned,ThroughputMaximization, Letter,Modelingairtoground,bor,kalantari,MozaffariIoT} was focused on performance analysis and did not consider prediction of user-centric patterns such as mobility nor does it study the use of UAVs for caching purposes. 
The prediction of the users' mobility patterns can enable the UAVs to effectively move and provide service for the ground users.
Moreover, when UAVs are considered within a CRAN system, the network must take into account the fact that the fronthaul links that connect the UAVs to the cloud will be capacity-limited. This is due to the fact that the bandwidth of the UAVs fronthaul links is limited. 
To overcome this limited-fronthaul capacity challenge, one can use content caching techniques to proactively download and cache content at the UAVs during off peak hours or when the UAVs are back at their docking stations. The use of caching enables the UAVs directly transmit the content to its requested user, thus reducing the fronthaul traffic load. 

The main contribution of this paper is to develop a novel framework that leverages user-centric information, such as content request distribution and mobility patterns, to effectively deploy cache-enabled UAVs while maximizing the users' QoE using a minimum total transmit power of the UAVs. The adopted QoE metric captures human-in-the-loop features such as transmission delay and the users' perceptions on the rate requirement, depending on their device type. In the proposed framework, the cloud can accurately predict the content request distribution and mobility patterns of each user. These predictions of user's behavior can then be used to find the optimal locations and content caching strategies for the UAVs. 
Unlike previous studies such as 
\cite{Bigdata,SoysaPredicting,NagarajaCaching,EchoStateNetworks2,ProactiveCaching2016} that predict the users' behavior using only one non-linear system, we propose a conceptor-based echo state network (ESN) approach to perform users' behavior prediction. Such an ESN model with conceptors enables the cloud to separate the users' behaviors into different patterns and learn these patterns independently, thus leading to a significant improvement in the accuracy of predictions. Moreover, unlike previous studies such as \cite{Mozaffari2016Unmanned, ThroughputMaximization, Letter,Modelingairtoground,kalantari,bor,MozaffariIoT} that consider the deployment of the UAVs assuming static users, we study the deployment of cache-enabled UAVs in CRANs with mobile users. In the proposed CRANs model, we derive the optimal user-UAV association, the optimal locations of the UAVs as well as the
content to cache at the UAVs.  
\emph{To our best knowledge, this work is the first to analyze the use of caching at the level of UAVs, given ESN-based predictions on the users' behavior}. To evaluate the performance of the proposed approach, we use real data from Youku for content requests as well as realistic measured mobility data from the Beijing University of Posts and Telecommunications for mobility simulations. Simulation results show that the proposed algorithm can yield 40\% gain in terms of the average transmit power of the UAVs compared to a baseline algorithm without cache. Moreover, the proposed algorithm can also yield 61\% gain in terms of the percentage of the users with satisfied QoE compared to a benchmark scenario without UAVs. 

The rest of this paper is organized as follows. The system model and problem formulation are presented in Section \uppercase\expandafter{\romannumeral2}. The conceptor ESN for content request distribution and mobility patterns predictions is proposed in Section \uppercase\expandafter{\romannumeral3}. The proposed approach for user-UAV association, content caching, and optimal location of each UAV is presented in Section \uppercase\expandafter{\romannumeral4}. In Section \uppercase\expandafter{\romannumeral5}, we provide numerical and simulation results. Finally, conclusions are drawn in Section \uppercase\expandafter{\romannumeral6}.  
     

\section{System Model and Problem Formulation}
\begin{figure}[!t]
  \begin{center}
   \vspace{0cm}
    \includegraphics[width=9cm]{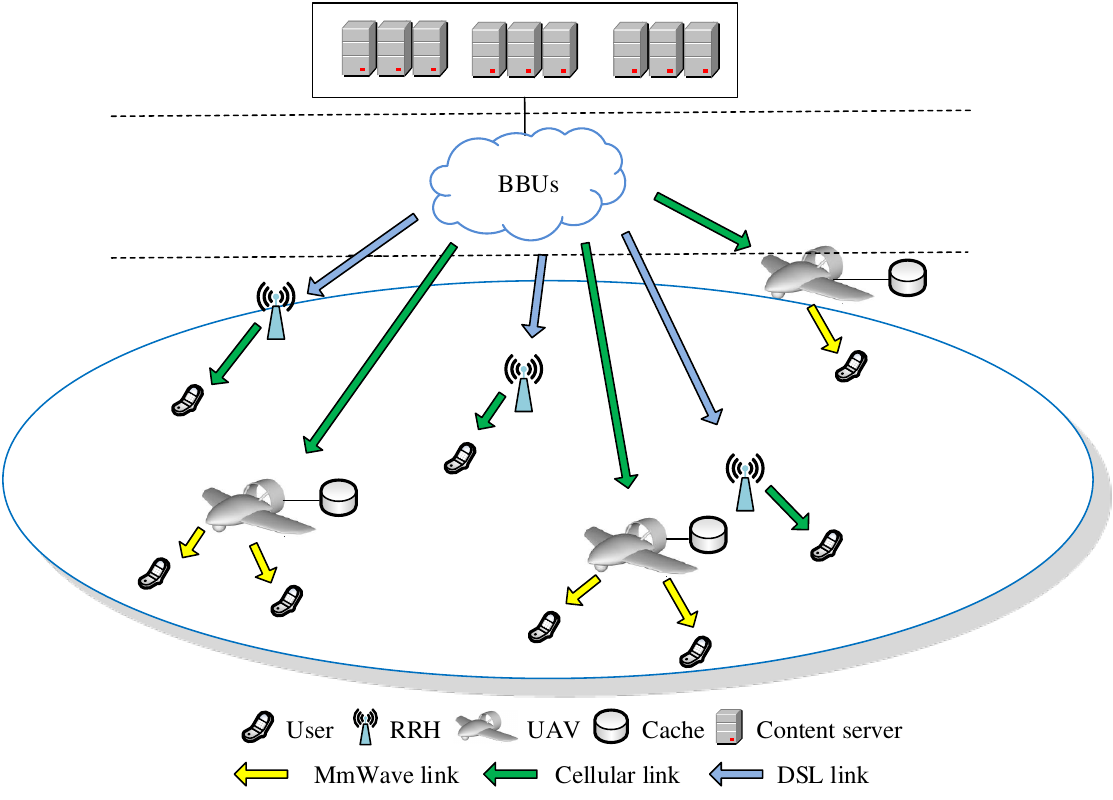}
    \vspace{-0.6cm}
    \caption{\label{CRAN} A CRAN with cache-enabled UAVs.}
  \end{center}\vspace{-1.2cm}
\end{figure}
\vspace{-0cm}
Consider the downlink of a CRAN system servicing a set $\mathcal{U}$ of $U$ mobile users via a set $\mathcal{R}$ of $R$ remote radio heads (RRHs) acting as distributed antennas.
The RRHs are grouped into $E$ clusters using \emph{K-mean clustering} approach \cite{H2008Clustering} so that zero-forcing beamforming (ZFBF)  \cite{Yoo2006On} 
can be used to service the users. In this system, a set $\mathcal{K} $ of $K$ UAVs equipped with cache storage units can be deployed to act as flying cache-enabled RRHs to serve the ground users along with the terrestrial RRHs. For the UAV-to-users communication links, since the high altitude of the UAVs can significantly reduce the blocking effect due to obstacles, we consider air-to-ground UAV transmissions using the millimeter wave (mmWave) frequency spectrum. Meanwhile, the terrestrial RRHs transmit over the cellular band and are connected to the cloud's pool of the baseband units (BBUs) via capacity-constrained, digital subscriber line (DSL) fronthaul links. Further, the cloud connects to the content servers via fiber backhaul links. The transmissions from the cloud to the UAVs occurs over wireless fronthaul links using the licensed cellular band. 
Consequently, the UAVs' wireless fronthaul links may interfere with the transmission links from the RRHs to the users when the user's requested content needs to be transmitted from the content server.

In our model, the content server stores a set ${\mathcal{N}}$ of all $N$ contents required by all users. The contents are of equal size $L$. Caching at the UAVs, referred to as ``UAV cache'' hereinafter, will be used to store the popular content that the  users request. By caching predicted content, 
the transmission delay from the content server to the UAVs can be significantly reduced as each UAV can directly transmit its stored content to the users.  Different from caching at the RRHs or BBUs, caching at UAVs allows servicing mobile users when their QoE requirement cannot be satisfied by the RRHs. 
We denote the set of $C$ cached contents in the storage units of UAV $k$ by $\mathcal{C}_k$, where $C \le N$ and $k \in \mathcal{K}$. {For simplicity, we assume that each user can request at most one content during each specified time slot $\tau$. We also let $\Delta\tau$ be the duration of time slot $\tau$ that also represents the maximum transmission duration of each content.} The maximum transmission duration $\Delta\tau$ is determined by the proposed algorithm in Section \ref{section3}. We assume that the content stored at the UAV cache will be refreshed every period that consists of $T$ time slots and this UAVs caching is performed at off peak hours when the UAVs return to their cloud-based docking stations for purposes such as battery charge.  
 Table \uppercase\expandafter{\romannumeral1} provides a summary of the notations used throughout this paper.

\begin{table}\footnotesize
  \newcommand{\tabincell}[2]{\begin{tabular}{@{}#1@{}}#2\end{tabular}}
\renewcommand\arraystretch{0.8}
 \caption{
    \vspace*{-0.5em}List of notations}\vspace*{-1.5em}
\centering  
\begin{tabular}{|c||c|c||c|}
\hline
Notation & Description & Notation & Description\\
\hline
$U$ & Number of users & $C$ & Number of the contents that are stored at UAV cache\\
\hline
$K$ & Number of UAVs & $F$ & Number of intervals in each time slot \\
\hline
 $R$ &  Number of RRHs  & $H$ & Number of time slots to collect user mobility \\
\hline
 $P_R$ & Transmit power of RRHs& $P _{t,ki}$& Transmitted power of UAV or RRH\\
\hline
$N$ &  Number of contents     &  $\tau$, $\Delta\tau$ & Time slot index, Time slot duration \\
\hline
$l_{t,ki}$ & Path loss of UAVs-users & $d_{t,ki}$ & Distance between RRHs or UAVs and users\\
\hline
$x_{\tau,k},y_{\tau,k},h_{\tau,k}$ & Coordinates of UAVs & $\delta_{S_i,n}$ & Rate requirement of device type \\
\hline
$L_{FS}$ &Free space path loss & $d_0$ & Free-space reference distance\\
\hline
$f_c$ & Carrier frequency &  $l _{t,ki}^F$& Path loss of fronthaul links\\
\hline
$\mu_\textrm{LoS}$, $\mu_\textrm{NLoS}$ & Path loss exponents & $\chi _{\sigma_\textrm{LoS}}, \chi _{\sigma_\textrm{NLoS}}$ & Shadowing random variable\\
\hline
$\gamma _{t,ki}^{\textrm{V}}, \gamma _{t,ki}^{\textrm{H}}$& SINR of user $i$ & $L_{t,k}^\textrm{LoS}$, $L_{t,k}^\textrm{NLoS}$ & LoS/NLoS path loss from the BBUs to UAV $k$\\
\hline
  $t$ & Small interval  & $l_{t,k}^\textrm{LoS}$, $l_{t,k}^\textrm{NLoS}$ & LoS/NLoS path loss from the to UAV $k$ to users\\
\hline
$c$ & Speed of light & $h_{t,ki}$ & Channel gains between the RRHs $k$ and user $i$\\
\hline
 $\bar D_{\tau,i,n}$ & Delay & $C_{\tau,ki}^F$& Fronthaul rate of UAV or RRH $k$\\
\hline
$C_{\tau,ki}^\textrm{V}$ & Rate of UAV-user link & $C_{\tau,qi}^\textrm{H}$& Rate of RRH-user link\\
\hline

$Q_{\tau,i,n}$& QoE of each user $i$& $T$& Number of time slots for caching update \\
\hline
$x_{t,i},y_{t,i}$ & Coordinates of users   & $P_B$ & Transmit power of the BBUs\\
\hline
\end{tabular}
 \vspace{-1cm}
\end{table} 

\subsection{Mobility Model}
In our system, we assume that the users can move continuously. 
In this case, we consider a realistic model for periodic, daily, and \emph{pedestrian mobility patterns} according to which each user will regularly visit a certain location of interest. For example, certain users will often go to the same office for work at the same time during weekdays. The locations of each user are collected by the BBUs once every $H$ time slots. Here, the duration of $H$ time slots is considered a period of one hour for each user. In addition, we assume that each user moves between two collected locations at a constant speed. The mobility pattern for each user will then be used to determine the content that must be cached as well as the optimal location of each UAV which will naturally impact the QoE of each user. 

In this model, the associations of the mobile users with the UAVs or the RRHs can change depending on the QoE requirement. 
Since the users are moving continuously, the locations of the UAVs must change  accordingly so as to serve the users effectively. However, for tractability, we assume that the UAVs will remain static during each content transmission. In essence, the UAVs will update their locations according to the mobility of the users after each content transmission is complete at a current location. 
 \vspace{-0.5cm}
\subsection{Transmission Model}
Next, we introduce the models for transmission links between BBUs and UAVs, UAVs and users, and RRHs and users. For ease of exposition, a time slot $\tau$ is discretized into $F$ equally spaced time intervals $t$, i.e., $\Delta\tau = Ft$. The time interval $t$ is chosen to be sufficiently small so that each user's location can be considered constant during $t$ as in \cite{ThroughputMaximization} and \cite{Abou2013Optimal}. 

\subsubsection{UAVs-Users Links} The mmWave propagation channel of the UAVs-user link  is modeled using the standard log-normal shadowing model of \cite{WirelessCommunications}. The standard log-normal shadowing model can be used to model the line-of-sight (LoS) and non-line-of-sight (NLoS) links by choosing specific channel parameters. Therefore, the LoS and NLoS path loss of UAV $k$ located at $\left(x_{\tau,k},y_{\tau,k},h_{\tau,k}\right)$ transmitting a content to user $i$ at interval $t$ of time slot $\tau$ is \cite{BroadbandMillimeterWave2013} (in dB):
\begin{equation}
\setlength{\abovedisplayskip}{3 pt}
\setlength{\belowdisplayskip}{3 pt}
l_{t,ki}^{\textrm{LoS}}\left( \boldsymbol{w}_{\tau,t,k},\boldsymbol{w}_{\tau,t,i}  \right) = {L_{FS}}\left( {{d_0}} \right) + 10\mu_{\textrm{LoS}}\log \left( {{d_{t,ki}}\left(  \boldsymbol{w}_{\tau,t,k},\boldsymbol{w}_{\tau,t,i}  \right)} \right) + {\chi _{\sigma_\textrm{LoS}} },
\end{equation}
\begin{equation}
l_{t,ki}^{\textrm{NLoS}}\left(\boldsymbol{w}_{\tau,t,k},\boldsymbol{w}_{\tau,t,i}  \right)\! = \!{L_{FS}}\left( {{d_0}} \right) \!+ \!10\mu_{\textrm{NLoS}}\log \left( {{d_{t,ki}}\left(  \boldsymbol{w}_{\tau,t,k},\boldsymbol{w}_{\tau,t,i}  \right)} \right)\! +\! {\chi _{\sigma_\textrm{NLoS}} },
\end{equation}
where $\boldsymbol{w}_{\tau,t,k}=\left[ {{x_{\tau,k}},{y_{\tau,k}},h_{\tau,k}} \right]$ is the coordinate of UAV $k$ during time slot $\tau$ with $h_{\tau,k}$ being the altitude of UAV $k$ at time slot $\tau$. Also, $\boldsymbol{w}_{\tau,t,k}=\left[ {{x_{t,i}},{y_{t,i}}} \right]$ is the time-varying coordinate of user $i$ at interval $t$. $L_{FS}\left(d_0\right)$ is the free space path loss given by $20\log \left( {{{d_0f_c4\pi } \mathord{\left/
 {\vphantom {{4\pi } c}} \right.
 \kern-\nulldelimiterspace} c}} \right)$ with $d_0$ being the free-space reference distance, $f_c$ being the carrier frequency and $c$ being the speed of light. ${{d_{t,ki}}\left( \boldsymbol{w}_{\tau,t,k},\boldsymbol{w}_{\tau,t,i} \right)}=\sqrt {\left(x_{t,i}-x_{\tau,k}\right)^2 + \left(y_{t,i}-y_{\tau,k}\right)^2 + h_{\tau,k}^2}$ is the distance between user $i$ and UAV $k$. $\mu_\textrm{LoS}$ and $\mu_\textrm{NLoS}$ are the path loss exponents for LoS and NLoS links. ${\chi _{\sigma_\textrm{LoS}} }$ and ${\chi _{\sigma_\textrm{NLoS}} }$ are the shadowing random variables which are, respectively, represented as the Gaussian random variables in dB with zero mean and $\sigma_\textrm{LoS}$, $\sigma_\textrm{NLoS}$ dB standard deviations. 
 
In our model, the probability of LoS connection depends on the environment, density and height of buildings, the locations of the user and the UAV, and the elevation angle between the user and the UAV. The LoS probability is given by \cite{Mozaffari2016Unmanned} and \cite{Modelingairtoground}:
\begin{equation}\label{eq:propability}
\setlength{\abovedisplayskip}{3 pt}
\setlength{\belowdisplayskip}{3 pt}
\Pr \left( {l_{t,ki}^\textrm{LoS}} \right) = {\left( {1 + X\exp \left( { - Y\left[ {\phi_t   - X} \right]} \right)} \right)^{ - 1}},
\end{equation}
where $X$ and $Y$ are constants which depend on the environment (rural, urban, dense urban, or others) and $\phi_t  = {\sin ^{ - 1}}\left( {{h_{\tau,k}}/d_{t,ki}\left( \boldsymbol{w}_{\tau,t,k},\boldsymbol{w}_{\tau,t,i}\right) } \right)$ is the elevation angle. Clearly, the average path loss from the UAV $k$ to user $i$ at interval $t$ is \cite{Modelingairtoground}:
\begin{equation}\label{eq:averl}
\setlength{\abovedisplayskip}{3 pt}
\setlength{\belowdisplayskip}{3 pt}
\bar l_{t,ki}\left( \boldsymbol{w}_{\tau,t,k},\boldsymbol{w}_{\tau,t,i}\right) = \Pr \left( {l_{t,ki}^\text{LoS}} \right) \times {l_{t,ki}^\text{LoS}}+ \Pr \left( {{l_{t,ki}^\textrm{NLoS}}} \right) \times {l_{t,ki}^\textrm{NLoS}},
\end{equation}
where $\Pr \left( {l_{t,ki}^\textrm{NLoS}}\right)=1-\Pr \left( {l_{t,ki}^\textrm{LoS}}\right)$.  
Based on the path loss, the average signal-to-noise ratio (SNR)
of user $i$ located at $\boldsymbol{w}_{\tau,t,i}$ from the associated UAV $k$ at interval $t$ is given by:
 \begin{equation}\label{eq:sir}
 \setlength{\abovedisplayskip}{3 pt}
\setlength{\belowdisplayskip}{3 pt}
{\gamma _{t,ki}^{\textrm{V}}} =\frac{{{P_{t,ki}}}}{{{{10}^{{{{\bar l_{t,ki}}\left( \boldsymbol{w}_{\tau,t,k},\boldsymbol{w}_{\tau,t,i} \right)} \mathord{\left/
 {\vphantom {{{l_{t,ki}}\left( \boldsymbol{V}_{\tau,t,ki} \right)} {10}}} \right.
 \kern-\nulldelimiterspace} {10}}}}{\sigma ^2}}},
\end{equation}
where $P_{t,ki}$ is the transmit power of UAV $k$ to the user $i$ at time $t$, and $\sigma^2$ is the variance of the Gaussian noise. We assume that the total bandwidth available for each UAV is $B_V$ which is equally divided among the associated users. The channel capacity between UAV $k$ and user $i$ for each content transmission will be ${C_{\tau,ki}^\textrm{V}}=\sum\limits_{t = 1}^{{F}} {\frac{{{B_V}}}{{{U_k}}}{{\log }_2}\left( {1 + {\gamma _{t,ki}^{\textrm{V}}}} \right)}$ where $U_k$ is the number of the users associated with UAV $k$.   

\subsubsection{BBUs-UAVs Ground-to-Air Links}

For the BBUs-UAVs (ground-to-air) link, we consider probabilistic LoS and NLoS links over the licensed band. Since the distance of the UAVs fronthaul link may be larger compared to the distance of the UAV-user link, the cellular band can provide a more reliable transmission and a smaller path loss compared to the mmWave channel. In such a model, NLoS links experience higher attenuations than LoS links due to the shadowing and diffraction loss. The LoS and NLoS path loss from the BBUs to UAV $k$ at time $t$ of time slot $\tau$ can be given by \cite{Mozaffari2016Unmanned}:
\begin{equation}
\setlength{\abovedisplayskip}{3 pt}
\setlength{\belowdisplayskip}{3 pt}
{L_{t,k}^\text{LoS}} ={{d_{t,ki}}\left( \boldsymbol{w}_{\tau,t,k},\boldsymbol{w}_{\tau,t,B} \right)}^{-\beta},
\end{equation}
\begin{equation}
\setlength{\abovedisplayskip}{3 pt}
\setlength{\belowdisplayskip}{3 pt}
{L_{t,k}^\text{NLoS}}=\eta{{d_{t,ki}}\!\left( \boldsymbol{w}_{\tau,t,k},\boldsymbol{w}_{\tau,t,B} \right)}^{-\beta},
\end{equation}
where $\boldsymbol{w}_{\tau,t,B}=\left[ {{x_{B}},{y_{B}}} \right]$ is the location of the BBUs, and $\beta$ is the path loss exponent. The LoS connection probability and the average SNR of the link between the BBUs and UAV $k$ can be calculated using (\ref{eq:propability})-(\ref{eq:sir}). 



\subsubsection{RRHs-Users Links} In our model, RRHs are grouped into $E$ clusters. Then, the RRHs in each cluster use ZFBF to improve the users' rates. The received signals of the users
associated with RRHs cluster $q$ at interval $t$ is:
\begin{equation}
 \setlength{\abovedisplayskip}{3 pt}
\setlength{\belowdisplayskip}{3 pt}
{\boldsymbol{b}_{t,q}} =\sqrt{P_R} {\boldsymbol{H}_{t,q}}{\boldsymbol{F}_{t,q}}{\boldsymbol{a}_{t,q}} +\boldsymbol{n},
\end{equation}
where ${\boldsymbol{H}_{t,q}} \in {\mathbb{R}^{{U_q} \times R_q}}$ is the path loss matrix with $U_q$ being the number of users
associated with RRH cluster $q$, and $R_q$ is the number of RRHs' antennas. $P_R$ is the transmit power of each RRH which is assumed to be equal for all RRHs. ${\boldsymbol{a}_{t,q}} \in {\mathbb{R}^{{U_q} \times 1}}$ is the transmitted content at interval $t$ and ${\boldsymbol{n}_{t,q}} \in {\mathbb{R}^{{U_q} \times 1}}$ is the noise power. Also, ${\boldsymbol{F}_{t,q}} = \boldsymbol{H}_{t,q}^{{\rm H}}\left( \boldsymbol{H}_{t,q}\boldsymbol{H}_{t,q}^{{\rm H}}\right)^{-1} \in {\mathbb{R}^{{R_q} \times {U_q}}}$ is the beamforming matrix \cite{Somekh2009Cooperative}. We also assume that the bandwidth of each user associated wth the RRHs is $B$. Then, the received signal-to-interference-plus-noise-ratio (SINR) of user $i$ in cluster $\mathcal{M}_q$ at interval $t$ will be:  
\begin{equation}\label{eq:sinrR}
 \setlength{\abovedisplayskip}{3 pt}
\setlength{\belowdisplayskip}{3 pt}
{\gamma _{t,qi}^{\textrm{H}}} = \frac{{P_R{{\left\| {{\boldsymbol{h}_{t,qi}}{\boldsymbol{f}_{t,qi}}} \right\|}^2}}}{{\underbrace {\sum\limits_{j = 1,j \ne q}^N{\sum\limits_{l \in \mathcal{M}_j  }\sum\limits_{u \in \mathcal{U}_j} {P_R{{\left\| {{{{\boldsymbol{h}_{t,li}}{\boldsymbol{f}_{t,lu}}}}} \right\|}^2}} }}_{\text{other cluster RRHs interference}}
+\underbrace {{{P_B}{g_{t,Bi}}} {d_{t,Bi}}{\left( {{x_{B}},{y_{B}},{x_{t,i}},{y_{t,i}}} \right)^{ - \beta }}}_{\text{wireless fronthaul interference}}
+{\sigma ^2}}},
\end{equation}
where $\mathcal{M}_j$ is the set of the RRHs in group $j$, $\mathcal{U}_j$ is the set of the users associated with the RRHs in group $j$, $\boldsymbol{h}_{t,qi} \in {\mathbb{R}^{1 \times R_q}}$ is the channel gain between the RRHs in cluster $\mathcal{M}_q$ and user $i$ with $h_{t,qi}=g_{t,qi}d_{t,qi}\left(x_i,y_i\right)^{-\beta}$, $g_{t,qi}$ is the Rayleigh fading parameter at interval $t$, and $d_{t,qi}\left(x_i,y_i\right)=\sqrt {{{\left( {{x_{t,q}} - {x_{t,i}}} \right)}^2} + {{\left( {{y_{t,q}} - {y_{t,i}}} \right)}^2}} $ is the distance between RRH $q$ and user $i$ at interval $t$. $\boldsymbol{f}_{t,qi} \in {\mathbb{R}^{R_q \times 1}}$ is the beamforming vector. Given (9), the channel capacity between RRH cluster $\mathcal{M}_q$ and user $i$ for each content transmission is: 
\begin{equation}\label{eq:rate}
\setlength{\belowdisplayskip}{2 pt}
{C_{\tau,qi}^{\textrm{H}}}=\sum\limits_{t = 1}^{{F}} {B{{\log }_2}\left( {1 + {\gamma _{t,qi}^{\textrm{H}}}} \right)}. 
\end{equation}


\subsection{Quality-of-Experience Model}
Given the proposed models in the previous subsections, here, we present the QoE model for each user. The \emph{quality-of-experience} of each user is formally defined as a concrete human-in-the-loop metric that captures each user's data rate, delay, and device type.

\subsubsection{Delay}In the considered CRAN system, contents can be transmitted to the users via three types of links: (a) content server-BBUs-RRHs-user, (b) content server-BBUs-UAV-user, and (c) UAV cache-user. The backhaul link connecting the cloud to the core network is assumed to be fiber and, therefore, its delay is neglected. We assume that the capacity of the wired fronthaul links between the BBUs and the RRHs is limited to a maximum rate of $v_F$ for all users. Consequently, the fronthaul rate for each user receiving a content from the RRHs will be ${v_{FU}} = {{{v_F}} \mathord{\left/
 {\vphantom {{{v_F}} {{N_{FR}}}}} \right.
 \kern-\nulldelimiterspace} {{N_{FR}}}}$ with $N_{FR}$ being the number of the users receive contents from the RRHs. Thus, the delay of a user $i$ receiving content $n$ over the three types of links at each time slot $\tau$ can be written as:
\begin{equation}\label{eq:delay}
{D_{\tau ,i,n}} = \left\{ {\begin{array}{*{20}{c}}
{\begin{array}{*{20}{c}}
{\frac{L}{{{v_{FU}}}} + \frac{L}{{{C_{\tau ,qi}^{\textrm{H}}}}}}, &{\;\;\;\;\;\;\;\;\;\text{link} \left( a \right),}
\end{array}}\\
{\begin{array}{*{20}{c}}
{\frac{L}{{C_{\tau ,k}^F}} + \frac{L}{{{C_{\tau ,ki}^\textrm{V}}}}}, &{\;\;\;\;\;\;\;\;\;\text{link} \left( b \right),}
\end{array}}\\
{\begin{array}{*{20}{c}}
{\frac{L}{{{C_{\tau ,ki}^\textrm{V}}}}}, &{\;\;\;\;\;\;\;\;\;\;\;\;\;\;\;\;\;\;\;\;\text{link}\left( c \right),}
\end{array}}
\end{array}} \right.
\end{equation}
where ${{C_{\tau ,k}^F}}$ is the rate of content transmission from the BBUs to UAV $k$ which is calculated analogously to (\ref{eq:sinrR}) and (\ref{eq:rate}). Next, we derive the lower bound on the delay that each user can tolerate for each content transmission.
\begin{proposition}\label{pro1}
\emph{The lower bound of the delay for each user $i$ receiving content $n$ are given by:}
\begin{equation}
\setlength{\abovedisplayskip}{0 pt}
\setlength{\belowdisplayskip}{0 pt}
\min \left\{ {\frac{L}{{{v_{F}}}},\frac{L}{{C_K^{\max }}}} \right\} \le {D_{\tau ,i,n}},
\end{equation}
\emph{where $ C_K^{\max }=B_V{\log _2}\left( {1 + \frac{{{P_{\max}}}}{{{{10}^{{{\left( {{L_{FS}}\left( {{d_0}} \right) + 10\mu_{\textrm{LoS}} \log \left( h_{\min} \right)}-4\sigma_{\textrm{LoS}} \right)} \mathord{\left/
 {\vphantom {{\left( {{L_{FS}}\left( {{d_0}} \right) + 10\mu \log \left( h \right)} \right)} {10}}} \right.
 \kern-\nulldelimiterspace} {10}}}}{\sigma ^2}}}} \right)$ with $P_{\max}$ being the maximum transmit power of each UAV, and $h_{\min}$ being the minimum altitude of the UAV.}
\end{proposition}  
\begin{proof} 
From (\ref{eq:delay}), we can see that the delay of link (b) is larger than that of link (c) and the minimum delay of the link (a) is ${L \mathord{\left/
 {\vphantom {L {{v_{FU}}}}} \right.
 \kern-\nulldelimiterspace} {{v_{F}}}}$. Hence, we only need to consider the delay values between ${L \mathord{\left/
 {\vphantom {L {{v_{FU}}}}} \right.
 \kern-\nulldelimiterspace} {{v_{FU}}}}$ and ${L \mathord{\left/
 {\vphantom {L {{C_{\tau ,ki}}}}} \right.
 \kern-\nulldelimiterspace} {{C_{\tau ,ki}^\textrm{V}}}}$, $k \in \mathcal{K}$. To maximize ${C_{\tau ,ki}^\textrm{V}}$, we consider ${{d_{t,ki}}\left( {\boldsymbol{w}_{\tau,t,k},\boldsymbol{w}_{\tau,t,i}} \right)=h}$, and $P_{t,ki}=P_{\max}$. Then, the rate of the UAV-user link ${ C_{\tau ,ki}^\textrm{V}}$ is given by: 
\begin{equation}\small\label{eq:proof11}
\begin{split}
{ C_{\tau ,ki}^\textrm{V}} &=F {B_V{{\log }_2}\left( {1 + \frac{{{P_{t,ki}}}}{{{{10}^{{{{\bar l_{t,ki}}\left( {\boldsymbol{w}_{\tau,t,k},\boldsymbol{w}_{\tau,t,i}} \right)} \mathord{\left/
 {\vphantom {{{l_{t,ki}^\textrm{LoS}}\left( {\boldsymbol{w}_{\tau,t,k},\boldsymbol{w}_{\tau,t,i}} \right)} {10}}} \right.
 \kern-\nulldelimiterspace} {10}}}}{\sigma ^2}}}} \right)} \\
 &\le F {B_V{{\log }_2}\left( {1 + \frac{{{P_{\max}}}}{{{{10}^{{{{l_{t,ki}^\textrm{LoS}}\left( {\boldsymbol{w}_{\tau,t,k},\boldsymbol{w}_{\tau,t,i}} \right)} \mathord{\left/
 {\vphantom {{{l_{t,ki}^\textrm{LoS}}\left( {\boldsymbol{w}_{\tau,t,k},\boldsymbol{w}_{\tau,t,i}} \right)} {10}}} \right.
 \kern-\nulldelimiterspace} {10}}}}{\sigma ^2}}}} \right)}\\
&  \mathop  \le \limits^{\left( a \right)} FB_V{\log _2}\left( {1 + \frac{{{P_{\max}}}}{{{{10}^{{{\left( {{L_{FS}}\left( {{d_0}} \right) + 10\mu_\textrm{LoS} \log \left( h \right) - {{{{4\sigma _\textrm{LoS}}}}}} \right)} \mathord{\left/
 {\vphantom {{\left( {{L_{FS}}\left( {{d_0}} \right) + 10\mu \log \left( h \right) + E\left[ {{\chi _{{\sigma _1}}}} \right]} \right)} {10}}} \right.
 \kern-\nulldelimiterspace} {10}}}}{\sigma ^2}}}} \right),\\
\end{split}
\end{equation}
where (a) follows from the fact that with a probability close to one (greater than $99.99\%$), the Gaussian random variable $\chi_{\sigma_\textrm{LoS}}$ will have a value larger than $-4\sigma_\textrm{LoS}$. From (\ref{eq:proof11}), we can see that, as $h$ increases, the capacity ${ C_{\tau ,ki}^\textrm{V}}$ decreases. Therefore, we set $h=h_{\min}$. This completes the proof.        
\end{proof}
From Proposition \ref{pro1}, we can see that the minimum delay of each user depends on the rate of the fronthaul links and the maximum transmit power of the UAVs. Therefore, we can improve the QoE of each user by adjusting the UAV's transmit power. In particular, as the number of users increases and the rate of fronthaul links decreases, the QoE requirement of users can be satisfied by adjusting the UAVs' transmit power.
{Note that, the upper bound of the delay $\Delta\tau$ is set by the system requirement.}      
Using the results of Proposition \ref{pro1}, we can categorize the sensitivity to the delay into five groups using the popular mean opinion score (MOS) model \cite{ContextAwareMitra} which is often used to measure the QoE of a wireless user. The mapping between delay and MOS model \cite{ContextAwareMitra} is given by: 
\begin{equation}\label{eq:delayD}
{\bar D_{\tau ,i,n}}=\frac {{\Delta\tau  - {D_{\tau ,i,n}}} } {{ {\Delta\tau  - \min \left\{ {\frac{L}{{{v_F}}},\frac{L}{{C_K^{\max }}}} \right\}} }},
 \end{equation}
 which is shown in Table \ref{tb:3}.


\subsubsection{Device Type} The screen size of each device type of the user will also affect the QoE perception of the user, especially for video-oriented applications. Indeed, users who own devices that have larger screens (such as tablets) will be more sensitive to QoE compared to those who own smaller devices (such as small smartphones). We capture the impact of the screen size of each user $i$ using a parameter $S_i$ that reflects the diameter length of the user's device. Typically, devices with a larger screen size, can display content at a higher resolution thus requiring a higher data rate. We assume that the rate requirement of user $i$ with device $S_i$ receiving a content $n$ at interval $t$ is $\delta _{S_i,n}=S_i{\hat C_{n}}$, where ${\hat C_{n}}$ is the rate requirement of each user receiving content $n$. The mapping from the rate requirement of user device to the MOS model can be written as:
\begin{equation}\label{eq:device}
V_{t,i}=\left\{ {\begin{array}{*{20}{c}}
{1,\;\;\;j \ge {\delta _{S_i,n}}},\\
{0,\;\;\;j < {\delta _{S_i,n}}},\\
\end{array}} \right. 
\end{equation}
where $j \in \left\{ {C_{t,ki}^\textrm{V}},{C_{t,qi}^\textrm{H}} \right\}$. From (\ref{eq:device}), we can see that the device type score can be equal to 1 or 0 meaning that the MOS can be ``Excellent'' or ``Poor''. 
 The QoE of each user $i$ receiving content $n$ at time slot $\tau$ can be given by \cite{ContextAwareMitra}:
\begin{equation}
\setlength{\abovedisplayskip}{3pt}
\setlength{\belowdisplayskip}{3 pt}
 Q_{\tau,i,n}= {{{\zeta_1} {{\bar D_{\tau ,i,n}} } + {\zeta_2}\sum\limits_{t = 1}^{{F}} {{V_{t,i}} }}}, 
\end{equation}
where $q_1$ and $q_2$ are weighting parameters with $\zeta_1+\zeta_2 = 1$.

\begin{table}\footnotesize
  \newcommand{\tabincell}[2]{\begin{tabular}{@{}#1@{}}#2\end{tabular}}
\renewcommand\arraystretch{0.7}
 \caption{
    \vspace*{-0.5em}Mean Opinion Score Model \cite{ContextAwareMitra}}\vspace*{-1.5em}
\centering  
\begin{tabular}{|c|c|c|c|c|c|c|}
\hline
QoE & Poor & Fair & Good & Very Good & Excellent \\
\hline
Interval scale & 0-0.2 & 0.2-0.4 & 0.4-0.6 & 0.6-0.8 &0.8-1\\
\hline
\end{tabular}\label{tb:3}
 \vspace{-0.9cm}
\end{table}

\subsection{Problem Formulation}
Here, we first find the minimum rate required to meet the QoE requirement of each user associated with the UAVs. Next, we determine the minimum transmit power of each UAV required to meet the QoE threshold of the associated users. Finally, we formulate the minimization problem. 
From Table \ref{tb:3}, we can see that, for $0.8 \le \bar D_{\tau,i,n} \le 1$, the MOS of  delay will be ``Excellent'', which means that the delay is minimized. In this case, $\bar D_{\min}=0.8$ is the minimum value that maximizes the delay component of user $i$'s QoE, during the transmission of a given content $n$. We define the rate that achieves the optimal delay as the delay rate requirement and also, define the rate that meets the rate requirement of device as the device rate requirement.
Consider the transmission between a UAV $k$ located at $\boldsymbol{w}_{\tau,t,k}$ and a user $i$ located at coordinates $\boldsymbol{w}_{\tau,t,i}$. 
From (\ref{eq:delay}), the delay rate requirement for UAV $k$ transmitting content $n$ to user $i$ at time slot $\tau$ is: 
 \begin{equation}\label{eq:ctaoki}
 C_{\tau ,ki,n}^R = \left\{ {\begin{array}{*{20}{c}}
{\frac{L}{{{}\left( {\Delta\tau  - {{\bar D}_{\min }}\left( {\Delta\tau  - \min \left\{ {\frac{L}{{{v_F}}},\frac{L}{{C_k^{\max }}}} \right\}} \right) - \frac{L}{{C_{\tau ,k}^F}}} \right)}},n \notin {\mathcal{C}_k},}\\
{\frac{L}{{{}\left( {\Delta\tau  - {{\bar D}_{\min }}\left( {\Delta\tau  - \min \left\{ {\frac{L}{{{v_F}}},\frac{L}{{C_k^{\max }}}} \right\}} \right)} \right)}}, \:\:\:\:\:\:\:\:\:\:\:n \in {\mathcal{C}_k}.}
\end{array}} \right.
 \end{equation}
From (\ref{eq:ctaoki}), we can see that, by storing content $n$ at cache of UAV $k$, the delay rate requirement for minimizing delay decreases.

Let ${\delta _{{S_i},n}}$ be the device rate requirement of user $i$ associated with a UAV. Clearly, the QoE is maximized when ${C_{t ,ki}^\textrm{V}} \ge \max \left\{ {{{C_{\tau ,ki,n}^R} \mathord{\left/
 {\vphantom {{C_{\tau ,ki}^R} {{F}}}} \right.
 \kern-\nulldelimiterspace} {{F}}},{\delta _{{S_i},n}}} \right\}$. Hence, the minimum rate required to maximize the user's QoE is $\delta_{i,n}^R=\max \left\{{{{{C_{\tau ,ki,n}^R} \mathord{\left/
 {\vphantom {{C_{\tau ,ki,n}^R} {{F}}}} \right.
 \kern-\nulldelimiterspace} {{F}}}},{\delta _{{S_i},n}}} \right\}$. Based on (\ref{eq:sir}), the minimum transmit power needed to guarantee the QoE requirement of user $i$ receiving content $n$ at interval $t$ is:
\begin{equation}\label{eq:Pki}
{P_{t,ki}^{\min}\left(\boldsymbol{w}_{\tau,t,k}, \delta_{i,n}^R, n\right)}\! =\! \left( {{2^{{\delta_{i,n}^RU_k  \mathord{\left/
 {\vphantom {\delta  B}} \right.
 \kern-\nulldelimiterspace} B_V}}} \!- \!1} \right)\!{\sigma ^2}
{{10}^{{{\bar l_{t,ki}\left( \boldsymbol{w}_{\tau,t,k},\boldsymbol{w}_{\tau,t,i} \right) } \mathord{\left/
 {\vphantom {{l_{t,ki}\left( {{x_{\tau,k}},{y_{\tau,k}},h,{x_{t,i}},{y_{t,i}}} \right) } {10}}} \right.
 \kern-\nulldelimiterspace} {10}}}}. 
\end{equation}
From (\ref{eq:Pki}), we can see that the minimum transmit power of UAV $k$ transmitting content $n$ to user $i$ depends on the UAV's location, the rate needed to satisfy the QoE requirement of user $i$, and the transmitted content $n$.     

Given this system model, our goal is to find an effective deployment of cache-enabled UAVs to enhance the QoE of each user while minimizing the transmit power of the UAVs. This problem involves predicting the content request distribution and periodic locations for each user, finding the optimal contents to cache at the UAVs, determining the users' associations and adjusting the locations\footnote{\vspace{-0.1cm}{Typically, the speed of a UAV can reach up to 30 m/s while the average speed of each pedestrian ground user is less than 2 m/s. Therefore, in our model, we ignore the time duration that each UAV uses to change its location.}}, and transmit power of the UAVs. This problem can be formulated as follows:
\addtocounter{equation}{0}
\begin{equation}\label{eq:sum}
\begin{split}
&\mathop {\min }\limits_{{\mathcal{C}_{k}},\mathcal{U}_{\tau,k}, \boldsymbol{w}_{\tau,t,k}}{\sum\limits_{\tau = 1}^{{T}}\sum\limits_{k \in \mathcal{K}}\sum\limits_{i \in {\mathcal{U}_{\tau,k}}} \sum\limits_{t = 1}^{{F}} {{{\!P_{\tau,t,ki}^{\min}\left(\boldsymbol{w}_{\tau,t,k},\delta_{i,n}^R, n_{\tau,i}\right)}}}},
\end{split}
\end{equation}
\vspace{-0.5cm}
\begin{align}\label{c1}
\;\;\;\;\;\text{s. t.}\;\; &\scalebox{1}{$ 
{h_{\min }} \le h_{\tau,k},k \in \mathcal{K},$}
\tag{\theequation a}\\
&\scalebox{1}{$ i \ne j, i, j \in {{\mathcal{C}_k}},\mathcal{C}_k \subseteq \mathcal{N},k \in \mathcal{K},$} \tag{\theequation b}\\
&\scalebox{1}{$0<{P_{\tau,t,ki}^{\min}} \le P_{\max},i \in \mathcal{U},k \in \mathcal{K}$,} \tag{\theequation c}
\end{align}
where $P_{\tau,t,ki}^{\min}$ is the minimum transmit power of UAV $k$ to user $i$ at interval $t$ during time slot $\tau$.
$n_{\tau,i}$ is the content that user $i$ requests at time slot $\tau$, $\mathcal{U}_{\tau,k}$ is the set of the users that are associated with UAV $k$ at time slot $\tau$. $h_{\min}$ is the minimum altitude that each UAV can reach at time slot $\tau$. Here, (\ref{eq:sum}b) captures the fact that each cache storage unit at the UAV stores a single, unique content, and (\ref{eq:sum}c) indicates that the transmit power of the UAVs should be minimized. 
Since the problem as per (\ref{eq:sum}) is to satisfy the rate needed for meeting each user's QoE requirement during the next time period $T$, the predictions of users behavior will directly impact the solution. From (\ref{eq:sum}), we can see that the prediction of the users' mobility patterns enable the BBUs to find the optimal locations of the UAVs. Moreover, by predicting the users' content request distribution the BBUs can determine the most popular content to cache at the UAVs. 


\section{Conceptor Echo State Networks for Content and Mobility Predictions}\label{section2}
In this section, we propose a prediction algorithm using the framework of ESN with conceptors, to find the users' content request distributions and their mobility patterns.
The predictions of the users' content request distribution and their mobility patterns will then be used in Section \ref{section3} to find the user-UAV association, optimal locations of the UAVs and content caching at the UAVs.   
Echo state networks are a special type of recurrent neural networks designed for performing non-linear systems forecasting \cite{Harnessing}. The ESN architecture is based on a randomly connected recurrent neural network, called reservoir, which is driven by a temporal input. The state of the reservoir is a rich representation of the history of the inputs so that a simple linear combination of the reservoir units is a good predictor of the future inputs. In our model, the reservoir will be combined with the input to store the users' context information and will also be combined with the trained output matrix to output the predictions of the users' content request distribution and mobility patterns.
Here, a user's \emph{context} is defined as the current state and attribute of a user including time, week, gender, occupation, age, and device type (e.g., tablet or smartphone). 
Therefore, an ESN-based approach can use the users' context to predict the corresponding behavior such as content request and mobility. 

Compared to traditional neural network and deep learning approaches such as in \cite{Nguyen2012Extracting}, an ESN-based approach can quickly learn the mobility pattern and content request distribution without requiring significant training data due to the use of the echo state property. However, {traditional ESN-based prediction algorithms such as in \cite{EchoStateNetworks2} can be trained to predict only one mobility pattern for each user. 
In particular, to predict the weekly mobility pattern of each user using the traditional ESN approach, the users' context information for an entire week need to be used as input of the ESNs that act as one non-linear system.} In this conventional ESN approach, it is not possible to separate the users' contexts in a week into several days and train the ESNs to predict the user's mobility in each day with one specific non-linear system.       
To enable the ESN algorithm to predict the user's mobility pattern and content request distribution with various non-linear systems, the notion of a \emph{conceptor} as defined in \cite{Jaeger2014Controlling}, is an effective solution that allows characterizing the ESN's reservoir. Conceptors enable an ESN to perform a large number of mobility and content request patterns predictions. Moreover, new patterns can be added to the reservoir of the ESN without interfering with previously acquired ones. 
For each ESN algorithm, an ESN can record a limited number of history input data due to the echo state property of each ESN. Consequently, the learned pattern will be removed as the recorded input data is updated. Here, we call the ability of recording a limited number of history input data as the \emph{memory} of the ESN's reservoir. The idea of a conceptor can be used to allocate any free memory of an ESN's reservoir to the new learned patterns of the mobility and content request distribution. 

Next, we first introduce the components of a conceptor ESN-based prediction algorithm. Then, we formulate the conceptor ESN algorithm to predict the content request distribution and mobility patterns of the users. 
\subsection{Conceptor ESN Components}
The conceptor ESN-based prediction approach  consists of five components: a) agents, b) input, c) output, d) ESN model, and e) conceptor. Since the content request and mobility pattern are user-specific, we design the specific components for the algorithms of the content request distribution and mobility pattern predictions, separately. 
 
\subsubsection{Content request distribution prediction} The content request distribution prediction algorithm has the following components:    

$\bullet$ \emph{Agent}: The agent in our ESNs is the cloud. Since each ESN scheme typically performs a content request distribution prediction for just one user, the cloud's BBUs must implement $U$ conceptor ESN algorithms.  

$\bullet$ \emph{Input:} The conceptor ESN takes input as a vector $\boldsymbol{x}_{t,j}=\left[ {x_{tj1}, \cdots , x_{tjN_x}} \right]^{\mathrm{T}}$ that represents the context of user $j$ at time $t$ which includes gender, occupation, age, and device type (e.g., tablet or smartphone). Here, $N_x$ is the number of properties that constitute the context information of user $j$. The vector $\boldsymbol{x}_{t,j}$ is then used to determine the content request distribution ${\boldsymbol{y} _{t,j}}$ for user $j$. Note that, the input of the ESNs is the information related to the users' content requests. Our goal is to predict the content request distribution using the context of each user.  

$\bullet$ \emph{Output:} The output of the content request distribution prediction ESN at time $t$ is a vector of probabilities $\boldsymbol{y}_{t,j}= \left[ {{p_{tj1}},{p_{tj2}}, \ldots ,{p_{tjN}}} \right]$ that represents the probability distribution of content request of user $j$ with $p_{tjn}$ being the probability that user $j$ requests content $n$ at time $t$. 

$\bullet$ \emph{ESN Model:} An ESN model for each user $j$ can find the relationship between the input $\boldsymbol{x}_{t,j}$ and output $\boldsymbol{y}_{t,j}$, thus building the function between the user's context and the content request distribution. Mathematically, the ESN model consists of the output weight matrix $\boldsymbol{W}_j^{\alpha,\textrm{out}} \in {\mathbb{R}^{N \times N_w}}$ and the dynamic reservoir containing the input weight matrix $\boldsymbol{W}_j^{\alpha,\textrm{in}} \in {\mathbb{R}^{N_w \times N_x}}$, and the recurrent matrix $\boldsymbol{W}_j^\alpha \in {\mathbb{R}^{N_w \times N_w}}$ with $N_w$ being the number of the dynamic reservoir units.    
For each user $j$, the dynamic reservoir will be combined with the input $\boldsymbol{x}_{t,j}$ to store the history context of user $j$. The output weight matrix $\boldsymbol{W}_j^{\alpha, \textrm{out}}$ with the reservior is trained to approximate the prediction function. The ESN model of user $j$ is initially randomly generated following a uniform distribution. To ensure that the reservoir has the echo state property, $\boldsymbol{W}_j^\alpha$ is defined as a sparse matrix with a spectral radius less than one \cite{APractical}.

$\bullet$ \emph{Conceptors:} For content request distribution prediction, we collect the users' context information and the corresponding content requests during the same time slots for different weeks to train one content request distribution. 
We refer to each content request distribution as one prediction pattern. Given a sequence of the reservoir states $\boldsymbol{v}_{j}^i=\left[\boldsymbol{v}_{1,j}^i, \ldots, \boldsymbol{v}_{t,j}^i\right]$ with $\boldsymbol{v}_{t,j}^i = {\left[ {v_{t,j1}^i, \ldots ,v_{t,j{N_w}}^i } \right]^{\rm T}}$ being the reservoir state of prediction pattern $i$ at time $t$ and the state correlation matrix $\boldsymbol{R}_{j}^i=\mathbb{E}\left[ {\boldsymbol{v}_{t,j}^i{{\left( {\boldsymbol{v}_{t,j}^i} \right)}^{\rm T}}} \right]$, the conceptor of prediction pattern $i$ will be \cite{Jaeger2014Controlling}:
\begin{equation}\label{eq:conceptor}
\setlength{\abovedisplayskip}{4 pt}
\setlength{\belowdisplayskip}{4 pt}
\boldsymbol{M}_j^i = \boldsymbol{R}_j^i{\left( {\boldsymbol{R}_j^i + {\chi ^{ - 2}}\boldsymbol{I}} \right)^{ - 1}},
\end{equation}
where $\chi$ is \emph{aperture} defined in \cite{Jaeger2014Controlling}. The aperture $\chi$ needs to be appropriately set for accurately learning several mobility patterns. When the aperture is small, the reservoir of the ESN slightly changes for learning each new pattern. However, for a large aperture, the reservoir of the ESN changes significantly.  

\subsubsection{Mobility pattern prediction}The components of mobility pattern prediction algorithm are:

$\bullet$ \emph{Agents}: The agents in our conceptor ESNs are the BBUs. Since each ESN scheme typically performs mobility prediction for only one user, the BBUs must also implement $U$ conceptor ESN algorithms. 

$\bullet$ \emph{Input:} $\boldsymbol{m}_{t,j}=\left[ {m_{tj1}, \cdots , m_{tjN_x+1}} \right]^{\mathrm{T}}$ represents the current location of user $j$ and the context of this user  at time $t$. Using input $\boldsymbol{m}_{t,j}$, the future locations of user $j$ can be predicted.

$\bullet$ \emph{Output:} $\boldsymbol{s}_{t,j}=\left[ {s_{tj1}, \cdots , s_{tjN_s}} \right]^{\mathrm{T}}$ represents the predicted locations of user $j$ in the next time slots, where $N_s$ is the number of locations in the next $N_s$ time duration $H$.

$\bullet$ \emph{ESN Model:} The ESN model of mobility prediction consists of the output weight matrix $\boldsymbol{W}_j^\textrm{out} \in {\mathbb{R}^{N_s \times N_w}}$, the dynamic reservoir containing the input weight matrix $\boldsymbol{W}_j^\textrm{in} \in {\mathbb{R}^{N_w \times N_x+1}}$, and the recurrent matrix $\boldsymbol{W}_j \in {\mathbb{R}^{N_w \times N_w}}$. The generation of the mobility prediction ESN model is similar to the one in the content request distribution prediction case.

$\bullet$ \emph{Conceptors:} For mobility pattern prediction, we consider each user's mobility in each day during one week as one prediction pattern. The expression of the conceptors is the same as the one for the content request distribution given in (\ref{eq:conceptor}).  
\subsection{Conceptor ESN Algorithm for Content and Mobility Predictions }
Here, we present the proposed conceptor ESN algorithm to predict the content request distribution and mobility. The proposed algorithm consists of two stages: training and prediction stages. 

\subsubsection{Training Stage} The dynamic reservoir state ${\boldsymbol{v}_{t,j}^i}$ of prediction pattern $i$ for user $j$ at time $t$ which is used to store the states of user $j$ is given by \cite{APractical}:
\begin{equation}\label{eq:state}
\setlength{\abovedisplayskip}{5 pt}
\setlength{\belowdisplayskip}{5 pt}
{\boldsymbol{v}_{t,j}^i} ={\mathop{f}\nolimits}\!\left( {\boldsymbol{W}_j^\alpha{\boldsymbol{v}_{t - 1,j}^i} + \boldsymbol{W}_j^{\alpha,\textrm{in}}{\boldsymbol{x}_{t,j}}} \right),
\end{equation}
where $f\!\left(x\right)=\frac{{{e^x} - {e^{ - x}}}}{{{e^x} + {e^{ - x}}}}$. Note that, we consider the input and corresponding prediction output as a training data. In this case, we use $N_{tr}$ training data that consists of $N_{tr}$ users' contexts and the corresponding content request to calculate the conceptors and train the output weight matrix $\boldsymbol{W}_j^{\alpha,\textrm{out}}$. Based on $N_{tr}$ training data and (\ref{eq:state}), the reservoir states before update for each prediction pattern $j$ is $\boldsymbol{v}_{\text{old},j}^i=\left[ 0, \boldsymbol{v}_{1,j}^i, \ldots, \boldsymbol{v}_{{N_{tr}-1},j}^i\right]$ and the updated reservoir states are $\boldsymbol{v}_{j}^i=\left[\boldsymbol{v}_{1,j}^i, \ldots, \boldsymbol{v}_{{N_{tr}},j}^i\right]$. The matrix $\boldsymbol{v}_{\text{old},j}^i$ will be used to train an \emph{input simulation matrix} $\boldsymbol{D}_j \in {\mathbb{R}^{N_w \times N_w}}$ and $\boldsymbol{v}_{j}^i$ will be combined with the updated reservoir states of other prediction patterns to train the output weight matrix.

Then, $\boldsymbol{D}_j$ will be combined with output weight matrix $\boldsymbol{W}_j^{\alpha,\textrm{out}}$ to predict the content request distribution pattern for each user. For each added learning pattern $i$ of each user $j$, the update of $\boldsymbol{D}_j$ will be\cite{Jaeger2014Controlling}:
\begin{equation}\label{eq:D}
\boldsymbol{D}_{j}=\boldsymbol{D}_{\text{old},j}+\boldsymbol{D}_{\text{inc},j}^{i},
\end{equation}
where $\boldsymbol{D}_{\text{inc},j}^{i}={\left( {{{\left( {{{\boldsymbol{S}{\boldsymbol{S}^{\rm T}}} \mathord{\left/
 {\vphantom {{\boldsymbol{S}{\boldsymbol{S}^{\rm T}}} {\left( {{N_{tr}} - 1} \right) + {\chi ^{ - 2}}I}}} \right.
 \kern-\nulldelimiterspace} { {{N_{tr}} }  + {\chi ^{ - 2}}\boldsymbol{I}}}} \right)}^\dag }{{\boldsymbol{S}{\boldsymbol{T}^{\rm T}}} \mathord{\left/
 {\vphantom {{S{T^{\rm T}}} { {{N_{tr}} } }}} \right.
 \kern-\nulldelimiterspace} {{{N_{tr}} }}}} \right)^{\rm T}}$ with $\boldsymbol{S} = {\boldsymbol{F}_j^{i-1}}\boldsymbol{v}_{\text{old},j}^i$ and $\boldsymbol{T} = \boldsymbol{W}_j^{\alpha ,\textrm{in}}\boldsymbol{x}_j^i -\boldsymbol{D}_{\text{old},j}\boldsymbol{v}_{\text{old},j}^i$. Here, $\boldsymbol{F}_j^{i-1}=\neg  \vee \left\{ {\boldsymbol{M}_j^1, \ldots , \boldsymbol{M}_j^{i - 1}} \right\}$ is the free memory of the reservoir with $\neg$ and $\vee$ being the boolean operators \cite{Jaeger2014Controlling}, and $\boldsymbol{x}_j^i=\left[\boldsymbol{x}_{1,j}^i, \ldots, \boldsymbol{x}_{{N_{tr}},j}^i\right]$ is the input sequences of prediction pattern $i$. During the learning of each pattern $i$ of user $j$, the conceptor $\boldsymbol{M}_j^i$ can be computed using (\ref{eq:conceptor}).         
 
In our proposed ESN algorithm, the output weight matrix $\boldsymbol{W}_j^{\alpha,\textrm{out}}$ is trained in an offline manner using ridge regression \cite{APractical} to approximate the prediction function which is given by:
\begin{equation}\label{eq:w2}
\setlength{\abovedisplayskip}{4 pt}
\setlength{\belowdisplayskip}{4 pt}
 \boldsymbol{W}_j^{\alpha, \textrm{out}}=\boldsymbol{y}_j{\boldsymbol{v}_j^{\rm T}}{\left({\boldsymbol{v}_j^{\rm T}}\boldsymbol{v}_j + {\lambda ^2}\boldsymbol{\rm I}\right)^{ - 1}},
 \end{equation}
where ${\boldsymbol{v}_j} = {\left[ {\boldsymbol{v}_{j}^1, \boldsymbol{v}_{j}^2,\ldots, \boldsymbol{v}_{j}^{N_M}} \right]^{\rm T}}$ with $\boldsymbol{v}_{j}^i=\left[\boldsymbol{v}_{1,j}^i, \ldots ,\boldsymbol{v}_{{N_{tr}},j}^i\right]$ being the reservoir state sequence of prediction pattern $i$ for user $j$, and $N_M$ being the number of the prediction patterns of each user's content request distribution. In (\ref{eq:w2}), $\boldsymbol{v}_{j}^i$ can also be used to calculate the conceptor $\boldsymbol{M}_{j}^i$ for prediction pattern $i$ of user $j$.

\subsubsection{Prediction Stage}
Based on the learning stage, we can use the input simulation matrix $\boldsymbol{D}_j$, conceptors $\boldsymbol{M}_{j}=\left[\boldsymbol{M}_{j}^1, \ldots ,\boldsymbol{M}_{j}^{N_M}\right]$, and output weight matrix $\boldsymbol{W}_j^{\alpha,\textrm{out}}$ to obtain the corresponding predictions. In the prediction stage, the reservoir state of pattern $i$ of user $j$ is \cite{Jaeger2014Controlling}:
\begin{equation}\label{eq:statep}
\setlength{\abovedisplayskip}{4 pt}
\setlength{\belowdisplayskip}{4 pt}
{\boldsymbol{v}_{t,j}^i} =\boldsymbol{C}_j^i{\mathop{f}\nolimits}\!\left( {\boldsymbol{W}_j^\alpha{\boldsymbol{v}_{t - 1,j}^i} + \boldsymbol{D}_j{\boldsymbol{v}_{t - 1,j}^i}} \right).
\end{equation}
From (\ref{eq:statep}), we can see that the conceptor of pattern $j$, $\boldsymbol{C}_j$, controls the update of the reservoir states. By changing the conceptor $\boldsymbol{C}_j$, the ESN can predict different patterns in one ESN architecture. The prediction of content request distribution $i$ for user $j$ can be given by:
\begin{equation}\label{eq:output}
\setlength{\abovedisplayskip}{3 pt}
\setlength{\belowdisplayskip}{3 pt}
\boldsymbol{y}_{t,j} = {\boldsymbol{W}_{j}^{\alpha,\textrm{out}}} {{\boldsymbol{v}_{t,j}^i}}.
\end{equation}
From (\ref{eq:statep}) and (\ref{eq:output}), we can see that the conceptor ESN algorithm exploits an input simulation matrix $\boldsymbol{D}_j$ to control the memory of ESN reservoir. The conceptor ESN algorithm for predicting the content request distribution of each user $j$ is shown in Table \ref{tab:algo}.

\begin{table}[!t]
  \centering
  \caption{
    \vspace*{-0.5em}Proposed Conceptor ESN Prediction Algorithm}\vspace*{-1em}
    \begin{tabular}{p{6.3in}}
      \hline \vspace*{-0.8em}
      \textbf{Inputs:}\,\,$N_{tr}$ training data, \vspace*{-0.5em}\\
\hspace*{1em}\textit{Initialize:}   \vspace*{-0.3em}
$\boldsymbol{W}_j^{\alpha, \textrm{in}}$, $\boldsymbol{W}_j^{\alpha}$, $\boldsymbol{W}_j^{\alpha,\textrm{out}}$, $\boldsymbol{y}_{j}=0$, $\boldsymbol{D}_j=0$.
\vspace*{-0cm}

\hspace*{0em} \textit{Training Stage:}\begin{itemize}\vspace*{-0.4em}
\item[] \hspace*{-0em} \textbf{for} each prediction pattern $i$ \textbf{do}.\vspace*{-0.4em}
\item[] \hspace*{0.5em} \textbf{if} reservoir memory space $\boldsymbol{F}_j^{i-1}>0$ \textbf{do}.\vspace*{-0.4em}
\item[] \hspace*{2em}(a) BBUs collect the reservoir states $\boldsymbol{v}_{\text{old},j}^i$ and $\boldsymbol{v}_{j}^i$ to update $\boldsymbol{D}_j$, using (\ref{eq:D}).\vspace*{-0.4em}
\item[] \hspace*{2em}(b) BBUs use the reservoir states $\boldsymbol{v}_{j}^i$ to calculate the conceptor $\boldsymbol{C}_j^i$ using (\ref{eq:conceptor}).\vspace*{-0.4em}
\item[] \hspace*{0.5em} \textbf{else}\vspace*{-0.4em}
\item[] \hspace*{2em}(c) increase reservoir weight matrix $\boldsymbol{W}_j^{\alpha}$, re-train all the prediction patterns.
\item[] \hspace*{0.5em} \textbf{end if}\vspace*{-0.4em}
\item[] \hspace*{0em} \textbf{end for}\vspace*{-0.4em}
\item[] \hspace*{0em}(c) BBUs collect the reservoir states for all patterns $\boldsymbol{v}_{j}$ to train $\boldsymbol{W}_j^{\alpha,\textrm{out}}$, using (\ref{eq:w2}).
\end{itemize}\vspace*{-0cm}\vspace*{-0.4em}
\hspace*{0em}\textit{Prediction Stage:}\begin{itemize}\vspace*{-0em}\vspace*{-0.4em}
\item[]  \hspace*{0em}(a) BBUs chooses the conceptor to obtain the corresponding reservoir state, using (\ref{eq:statep}).\vspace*{-0.4em}
\item[]  \hspace*{0em}(b) Get the prediction of content request distribution based on (\ref{eq:output}) .\vspace*{-0.4em}
\end{itemize}\vspace*{-0cm}\vspace*{-0.1em}
\hspace*{0em}\textbf{Output:}\,\, Prediction $\boldsymbol{y}_{t,j}$\vspace*{0em}\\
   \hline
    \end{tabular}\label{tab:algo}\vspace{-0.6cm}
\end{table}

As shown in Table \ref{tab:algo}, the proposed conceptor ESN algorithm can learn each prediction pattern by a unique non-linear system. This property of the proposed algorithm enables the ESNs to perform the users' behavior predictions using different non-linear systems during different time periods. Furthermore, using the proposed algorithm, one can have the information of the reservoir memory and extract a specific prediction pattern from the learned patterns.

\section{Optimal Location and Content Caching for UAVs}\label{section3}
In this section, we use the content request distribution and mobility patterns predictions resulting from the proposed conceptor ESN algorithm in Section \ref{section2} to solve the problem in (\ref{eq:sum}). In our model, a subset of the users selected by the BBUs are connected to the RRHs. The remaining users are clustered into $K$ clusters and each UAV provides service for one cluster. Based on the associations and predictions, we determine which contents to cache at each UAV and find the optimal location of each UAV. Finally, we analyze the implementation and complexity of the proposed algorithm. Fig. \ref{CRAN} summarizes the proposed framework that is used to solve the problem in (\ref{eq:sum}).\vspace{-0.2cm} 

\begin{figure}[!t]
  \begin{center}
   \vspace{0cm}
    \includegraphics[width=13cm]{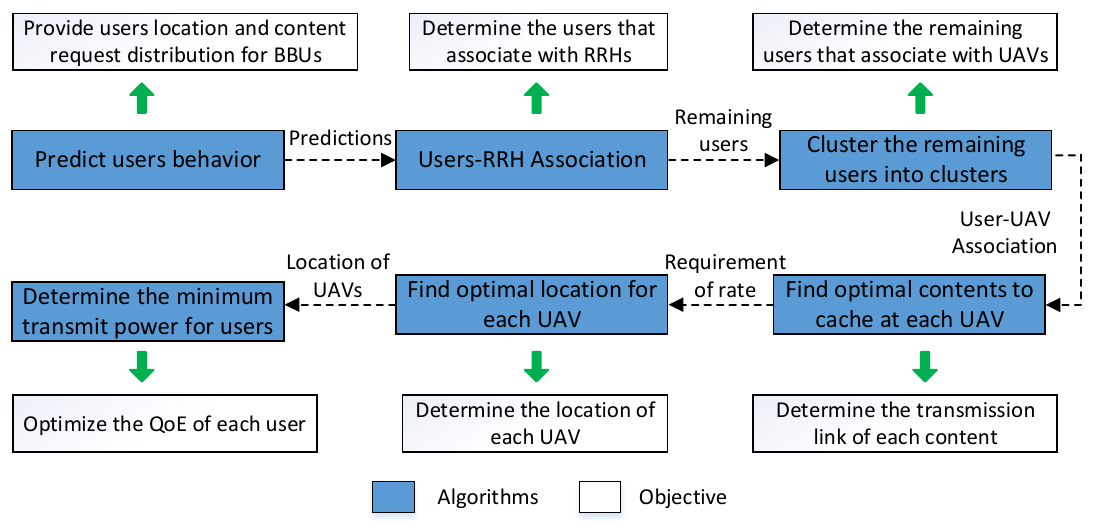}
    \vspace{-0.6cm}
    \caption{\label{CRAN} The procedure used for solving the optimization problem given in (\ref{eq:sum}).}
  \end{center}\vspace{-1cm}
\end{figure}
\vspace{-0cm} 

\subsection{Users-RRH Association}
We find the user-RRH association based on the predicted users' locations at the next time interval. 
Clearly, the prediction accuracy of the users' locations will directly affects the users association. 
A user is associated with RRHs if the following condition is satisfied: 
\begin{theorem}\label{th:1}
\emph{Given minimum $\bar D_{\min}$ and device screen size $S_i$ of each user $i$, user $i$ will be associated with a cluster $k$ of RRHs if the following rate requirement is satisfied:}
\begin{equation}
{C_{t ,qi}^\textrm{H}} \ge \max \left\{ {\frac{L}{{{F}\left( {\Delta\tau  - {{\bar D}_{\min }}\left( {\Delta\tau  - \min \left\{ {\frac{L}{{{v_F}}},\frac{L}{{C_k^{\max }}}} \right\}} \right) - \frac{L}{{{v_{FU}}}}} \right)}},{\delta _{{S_i},n}}} \right\}.
\end{equation}
\end{theorem} 
\begin{proof}  
Based on (\ref{eq:delayD}) and $\bar D_{\min}$, the delay is ${D_{\tau ,i,n}} =\Delta\tau  - {\bar D_{\min }}\left( {\Delta\tau  - \min \left\{ {\frac{L}{{{v_F}}},\frac{L}{{C_k^{\max }}}} \right\}} \right)$, and, hence, the delay rate requirement for RRH cluster $q$ transmitting content $n$ to user $i$ during time slot $\tau$ will be: 
\begin{equation}\label{eq:rater}
{C_{\tau,qi}^R} = \frac{L}{{\Delta\tau  - {{\bar D}_{\min }}\left( {\Delta\tau  - \min \left\{ {\frac{L}{{{v_F}}},\frac{L}{{C_k^{\max }}}} \right\}} \right) - \frac{L}{{{v_{FU}}}}}}.
\end{equation}
Therefore, the delay rate requirement during each interval is equal to ${{C_{\tau ,qi}^R} \mathord{\left/
 {\vphantom {{C_{\tau ,qi}^R} {{F}}}} \right.
 \kern-\nulldelimiterspace} {{F}}}$. Since the device rate requirement is $\delta_{S_i,n}$, the rate of RRH cluster $q$ transmitting content $n$ to user $i$, ${C_{\tau ,qi}^\textrm{H}}$ must satisfy ${C_{t ,qi}^\textrm{H}} \ge \max \left\{ {{{C_{\tau ,qi}^R} \mathord{\left/
 {\vphantom {{C_{\tau ,qi}^\textrm{R}} {{F}}}} \right.
 \kern-\nulldelimiterspace} {{F}}},{\delta _{{S_i},n}}} \right\}$. This completes the proof.  
\end{proof}
From Theorem \ref{th:1}, we can see that the user-RRH association depends on the fronthaul rate of each user, the delay rate requirement, and the device rate requirement. From (\ref{eq:rater}), we can see that the fronthaul rate of each user decreases as the number of the users associated with the RRHs increases. Clearly, the decrease of the fronthaul rate for each user will improve the delay rate requirement. 
\subsection{Optimal Content Caching for UAVs}      
In our model, the remaining users who are not associated with RRHs, will be served by the UAVs. In this case, the users-UAVs associations  need to be determined. To this end, we use K-mean clustering approach \cite{H2008Clustering} in which the users are clustered into $K$ groups.
By implementing the K-mean clustering approach, the users that are close to each other will be grouped into one cluster. Thereby, each UAV services one cluster and the user-UAV association will be determined. 
Then, based on the UAV association, we find the optimal contents to cache at each UAV. The content caching will reduce the transmission delay and, hence, decrease the delay rate requirement. 
From (\ref{eq:ctaoki}), we can see that, optimal contents to store at the UAV cache lead to maximum reduction of the UAV's transmit power. The reduction of UAV transmit power is caused by the decrease of the delay rate requirement. Let $\boldsymbol{p}_{j,i}={\left[ {{p_{j,i1}},{p_{j,i2}}, \ldots ,{p_{j,iN}}} \right]}$ be the content request distribution of user $i$ during period $j$ that consists of $H$ time slots. The optimal contents that will be stored at each UAV cache can be determined based on the following theorem.    
 \begin{theorem}\label{th:2}
\emph{The optimal set of contents $\mathcal{C}_{k}$ to cache at each UAV $k$ during period $T$ is:}
\begin{equation}\label{eq:content}
\mathcal{C}_{k}=\mathop {\arg\max }\limits_{{\mathcal{C}_k}} \sum\limits_{j = 1}^{{T \mathord{\left/
 {\vphantom {T H}} \right.
 \kern-\nulldelimiterspace} H}}\sum\limits_{\tau = 1}^{{H }}\sum\limits_{i \in {\mathcal{U}_{\tau,k}}}  { {\sum\limits_{n \in {\mathcal{C}_k}} \left({{p_{j,in}}\Delta {P_{j,\tau ,ki,n}} } \right)} },
\end{equation}
\emph{where {\small$\!\Delta {P_{j,\tau ,ki,n}}\! \!=\left\{ {\begin{array}{*{20}{c}}
{\!\!\!\!\!\!\!\!\!\!\!P_{\tau ,ki}^{\min }{{\left( {\boldsymbol{w}_{\tau,t,k}, C_{\tau ,ki}^R,n} \right)}_{n \notin {\mathcal{C}_k}}}\!\!\! \!\!\!-\!\! P_{\tau ,ki}^{\min }{{\left( {\boldsymbol{w}_{\tau,t,k}, C_{\tau ,ki}^R,n} \right)}_{n \in {\mathcal{C}_k}}},\;\; \;\;\;\;\;\; \;\; \;\;   \frac{{C_{\tau ,ki,n\left( {n \in {\mathcal{C}_k}} \right)}^R}}{{{F}}} \ge {\delta _{{S_i},n}},}\\
{\!\!\!\!P_{\tau ,ki}^{\min }{{\left(\! {\boldsymbol{w}_{\tau,t,k}, {\delta _{{S_i},n}},n}\! \right)}_{n \notin {\mathcal{C}_k}}}\!\! \!-\! \!P_{\tau ,ki}^{\min }{{\left(\! {\boldsymbol{w}_{\tau,t,k},C_{\tau ,ki}^R,n}\! \right)}_{n \in {\mathcal{C}_k}}}\!,\frac{{C_{\tau ,ki,n\left( {n \notin {\mathcal{C}_k}} \right)}^R}}{{{F}}} \!\!\ge {\delta _{{S_i},n}} \!\!\ge\!\! \frac{{C_{\tau ,ki,n\left( {n \in {\mathcal{C}_k}} \!\right)}^R}}{{{F}}},}\\
{\;\;\;\;\;\;\;\;\;\; \;\; \;\; \;\; \;\; \;\; \;\;\;\; \;\;\;\; \;\; \;\;   0,\;\; \;\; \;\; \;\; \;\; \;\;\;\;\;\;\;\;\;\;\;\;\;\;\;\; \;\; \;\; \;\; \;\; \;\; \;\; \;\;\;\; \;\; \;\; \;\; \;\; \;\; \;\; \;\;  \text{otherwise},}
\end{array}} \right.$} 
}
\end{theorem} 
\begin{proof}
Since the delay rate requirement ${C_{\tau ,ki,n}^R}$ depends on the contents at the UAV cache, it can be written as $\delta_{i,n}^R=\max \left\{{  \frac{{C_{\tau ,ki,n\left( {n \in {\mathcal{C}_k}} \right)}^R}}{{{F}}} , \frac{{C_{\tau ,ki,n\left( {n \notin {\mathcal{C}_k}} \right)}^R}}{{{F}}} ,{\delta _{{S_i},n}}} \right\}$. Let $P_{\tau,ki}^{\min }={\sum\limits_{t = 1}^{{F}} {P_{j,\tau,t,ki}^{\min }} }$. Then the reduction of UAV transmit power by content caching during time slot $\tau$ of period $j$ will be:
\begin{equation}\small\nonumber
\!\Delta {P_{j,\tau ,ki,n}}\! \!=\left\{ {\begin{array}{*{20}{c}}
{\!\!\!\!\!\!\!\!\!\!\!\!\!\!\!\!P_{\tau ,ki}^{\min }{{\left( {\boldsymbol{w}_{\tau,t,k},C_{\tau ,ki}^R,n} \right)}_{n \notin {\mathcal{C}_k}}} - P_{\tau ,ki}^{\min }{{\left( {\boldsymbol{w}_{\tau,t,k},C_{\tau ,ki}^R,n} \right)}_{n \in {\mathcal{C}_k}}},\;\; \;\; \;\;\;\; \;\; \;\;   \frac{{C_{\tau ,ki,n\left( {n \in {\mathcal{C}_k}} \right)}^R}}{{{F}}} \ge {\delta _{{S_i},n}},}\\
{P_{\tau ,ki}^{\min }{{\left( {\boldsymbol{w}_{\tau,t,k},{\delta _{{S_i},n}},n} \right)}_{n \notin {\mathcal{C}_k}}} - P_{\tau ,ki}^{\min }{{\left( {\boldsymbol{w}_{\tau,t,k},C_{\tau ,ki}^R,n} \right)}_{n \in {\mathcal{C}_k}}},\frac{{C_{\tau ,ki,n\left( {n \notin {\mathcal{C}_k}} \right)}^R}}{{{F}}} \ge {\delta _{{S_i},n}} \ge \frac{{C_{\tau ,ki,n\left( {n \in {\mathcal{C}_k}} \right)}^R}}{{{F}}},}\\
{\;\;\;\;\;\;\;\;\;\; \;\; \;\; \;\; \;\; \;\; \;\;\;\; \;\;\;\; \;\; \;\;   0,\;\; \;\; \;\; \;\; \;\; \;\;\;\;\;\;\;\;\;\;\;\;\;\;\;\; \;\; \;\; \;\; \;\; \;\; \;\; \;\;\;\; \;\; \;\; \;\; \;\; \;\; \;\; \;\;  \text{otherwise},}
\end{array}} \right.
\end{equation}
 Considering the fact that the content request distribution changes once every $H$ time slots, the power minimization problem for UAV $k$ during a period that consists of $H$ time slots is:
\begin{equation}\small
\begin{split} 
\mathop {\min }\limits_{{\mathcal{C}_{k}}}{\sum\limits_{\tau = 1}^{{T }}\sum\limits_{i \in {\mathcal{U}_{\tau,k}}}\!\! \!\!P_{\tau,ki}^{\min }}&=\mathop {\min }\limits_{{\mathcal{C}_{k}}}{
\sum\limits_{j = 1}^{{T \mathord{\left/
 {\vphantom {T H}} \right.
 \kern-\nulldelimiterspace} H}}\sum\limits_{\tau_j = 1}^{{H }}
 \sum\limits_{i \in {\mathcal{U}_{\tau,k}}}\!\! \!\!P_{\tau_j,ki}^{\min }}=\mathop {\min }  \limits_{{\mathcal{C}_{k}}} \sum\limits_{j = 1}^{{T \mathord{\left/
 {\vphantom {T H}} \right.
 \kern-\nulldelimiterspace} H}}{\sum\limits_{\tau = 1}^{{H }}\sum\limits_{i \in {\mathcal{U}_{\tau,k}}} \!\!\!P_{j,\tau,ki}^{\min }}\mathop = \limits^{\left( a \right)} \mathop {\max }  \limits_{{\mathcal{C}_{k}}} \sum\limits_{j = 1}^{{T \mathord{\left/
 {\vphantom {T H}} \right.
 \kern-\nulldelimiterspace} H}}{\sum\limits_{\tau = 1}^{{H }}\sum\limits_{i \in {\mathcal{U}_{\tau,k}}} \!\!\!\Delta {P_{j,\tau ,ki,n}}},\\
 &\mathop  = \limits^{\left( b \right)} \mathop {\max }\limits_{{\mathcal{C}_k}} \sum\limits_{j = 1}^{{T \mathord{\left/
 {\vphantom {T H}} \right.
 \kern-\nulldelimiterspace} H}}\sum\limits_{\tau = 1}^{{H }}\sum\limits_{i \in {\mathcal{U}_{\tau,k}}} \left( { {\sum\limits_{n \in {\mathcal{C}_k}} \!\left({{p_{j,in}}\Delta {P_{j,\tau ,ki,n}} } \right)} }+{ {\sum\limits_{n \notin {\mathcal{C}_k}} \!\left({{p_{j,in}}\Delta {P_{j,\tau ,ki,n}} } \right)} }\right),\\
&= \mathop {\max }\limits_{{\mathcal{C}_k}} \sum\limits_{j = 1}^{{T \mathord{\left/
 {\vphantom {T H}} \right.
 \kern-\nulldelimiterspace} H}}\sum\limits_{\tau = 1}^{{H }}\sum\limits_{i \in {\mathcal{U}_{\tau,k}}}  { {\sum\limits_{n \in {\mathcal{C}_k}} \!\left({{p_{j,in}}\Delta {P_{j,\tau ,ki,n}} } \right)} },
\end{split}
\end{equation}
where $\left(a \right)$ follows the fact that minimizing the transmit power of the UAVs is equivalent to maximizing the reduction of the UAVs' transmit power  caused by caching, and $\left(b\right)$ is obtained by computing the average power reduction using content request probability distribution of each user. This completes the proof. 
\end{proof}
 From Theorem \ref{th:2}, we can see that when the fronthaul rates of all users are the same, the transmit power reduction $\Delta {P_{j,\tau ,ki,n}}$ will be a constant. Subsequently, the optimal content caching becomes $\mathcal{C}_{k}=\mathop {\arg\max }\limits_{{\mathcal{C}_k}} \sum\limits_{j = 1}^{{T \mathord{\left/
 {\vphantom {T H}} \right.
 \kern-\nulldelimiterspace} H}}\sum\limits_{\tau = 1}^{{H }}\!\sum\limits_{i \in {\mathcal{U}_{\tau,k}}} \! { {\sum\limits_{n \in {\mathcal{C}_k}} {{p_{j,in}} }} }
$ which corresponds to the result given in \cite{EchoStateNetworks2}. From Theorem \ref{th:2}, we can see that the content caching depends on the pre-knowledge of users association as well as the content request distribution of each user. Therefore, by predicting the mobility pattern and content request distribution for each user, we can determine the optimal content to cache.

\subsection{Optimal Locations of UAVs}
Here, we determine the optimal UAVs' locations where the UAVs can serve their associated users using minimum transmit power. Once each UAV selects the suitable contents to cache, the transmission link (BBUs-UAV-user or UAV-user) for each content and the delay rate requirement $C_{\tau ,ki,n}^R$ in (\ref{eq:ctaoki}) are determined. In this case, the rate $\delta_{i,n}^R$ which is used to meet the QoE requirment of each user is also determined. 
Next, 
%
we derive a closed-form expression for the optimal location of UAV $k$ during time slot $\tau$ in two special cases.
 
 \begin{theorem}\label{th:3}
\emph{ To minimize the transmit power of UAV $k$, the optimal locations of UAV $k$ during time slot $\tau$ for cases: 
$a)$ UAV $k$ positioned at low altitudes compared to the size of its corresponding coverage, $h_{\tau,k}^2  \ll  \left(x_{t,i}-x_{\tau,k}\right)^2 + \left(y_{t,i}-y_{\tau,k}\right)^2$ and $\mu_\textrm{NLoS}=2$,
$b)$ UAV $k$ is placed at high altitudes compared to the size of its corresponding coverage, $h_{\tau,k}^2 \gg \left(x_{t,i}-x_{\tau,k}\right)^2 + \left(y_{t,i}-y_{\tau,k}\right)^2$, are given by:
\begin{equation}
\setlength{\abovedisplayskip}{3 pt}
\setlength{\belowdisplayskip}{3 pt}
x_{\tau,k}=\frac{{\sum\limits_{i \in \mathcal{U}_{\tau,k}} {\sum\limits_{t = 1}^{{F}} {{x_{t,i}}{\psi _{t,ki}}} } }}{{\sum\limits_{i \in \mathcal{U}_{\tau,k}} {\sum\limits_{t = 1}^{{F}} {{\psi _{t,ki}}} } }}, y_{\tau,k}=\frac{{\sum\limits_{i \in \mathcal{U}_{\tau,k}} {\sum\limits_{t = 1}^{{F}} {{y_{t,i}}{\psi _{t,ki}}} } }}{{\sum\limits_{i \in \mathcal{U}_{\tau,k}} {\sum\limits_{t = 1}^{{F}} {{\psi _{t,ki}}} } }},
\end{equation}
where $\psi_{t,ki}=\left( {{2^{{\delta_{i,n}^R  \mathord{\left/
 {\vphantom {\delta  B}} \right.
 \kern-\nulldelimiterspace} B}}} \!- \!1} \right)\!{\sigma ^2}
{{10}^{{{\left( {{L_{FS}}\left( {{d_0}} \right)+{\chi _{\sigma} }} \right)} \mathord{\left/
 {\vphantom {{\left( {{L_{FS}}\left( {{d_0}} \right) + 10\mu \log \left(  \right)} \right)} {10}}} \right.
 \kern-\nulldelimiterspace} {10}}}}$ with 
 $\sigma  = \left\{ {\begin{array}{*{20}{c}}
{{\sigma _{\textrm{NLoS}}},\;\; \textrm{for case} \left. a \right),}\\
{{\sigma _{\textrm{LoS}}},\;\;\;\; \textrm{for case} \left. b \right).}
\end{array}} \right.$ }
\end{theorem} 
\begin{proof}  
At very low altitudes, $h_{\tau,k}^2  \ll  \left(x_{t,i}-x_{\tau,k}\right)^2 + \left(y_{t,i}-y_{\tau,k}\right)^2$, $\frac{{{h_{\tau ,k}}}}{{{d_{t,ki}}\left( {\boldsymbol{w}_{\tau,t,k},\boldsymbol{w}_{\tau,t,i}} \right)}} \approx 0$ leading to ${\phi _t} = {0^\circ }$, and, consequently, $\Pr \left( {{l_{t,ki}^\textrm{NLoS}}} \right)=1$. Thus, we have $\bar l_{t,ki}\left(\boldsymbol{w}_{\tau,t,k}, \boldsymbol{w}_{\tau,t,i}\right) ={l_{t,ki}^\textrm{NLoS}}$ and (\ref{eq:Pki}) can be rewritten as $
{P_{\tau,t,ki}^{\min}}\! =\! \left( {{2^{{\delta_{i,n}^R  \mathord{\left/
 {\vphantom {\delta  B}} \right.
 \kern-\nulldelimiterspace} B}}} \!- \!1} \right)\!{\sigma ^2}
{{10}^{{{\left( {{L_{FS}}\left( {{d_0}} \right)+{\chi _{\sigma_\textrm{NLoS}} }} \right)} \mathord{\left/
 {\vphantom {{\left( {{L_{FS}}\left( {{d_0}} \right) + 10\mu \log \left(  \right)} \right)} {10}}} \right.
 \kern-\nulldelimiterspace} {10}}}} {{d_{t,ki}}\left( {\boldsymbol{w}_{\tau,t,k},\boldsymbol{w}_{\tau,t,i}} \right)^{\mu_\textrm{NLoS}}}$. 

  Now, we find the optimal location $\left(x_{\tau,k},y_{\tau,k}\right)$ of UAV $k$ during time slot $\tau$ 
 in order to minimize $\sum\limits_{i \in \mathcal{U}_{\tau,k}} \sum\limits_{t = 1}^{{F}}P_{\tau,t,ki}^{\min}$. In this case, the derivation of $\sum\limits_{i \in \mathcal{U}_{\tau,k}} \sum\limits_{t = 1}^{{F}}P_{\tau,t,ki}^{\min}$ with respect to $x_{\tau,k}$ is given by:
\begin{equation}\label{eq:de}\small
\frac{{\partial \sum\limits_{i \in \mathcal{U}_{\tau,k}} {\sum\limits_{t = 1}^{{F}} {P_{\tau,t,ki}^{\min }} } }}{{\partial\! {x_{\tau ,k}}}} = \frac{{\sum\limits_{i \in \mathcal{U}_{\tau,k}} {\sum\limits_{t = 1}^{{F}} {\partial P_{\tau,t,ki}^{\min }} } }}{{\partial {x_{\tau ,k}}}}=\sum\limits_{i \in \mathcal{U}_{\tau,k}} {\sum\limits_{t = 1}^{{F}} \mu_\textrm{NLoS} \left( {{x_{\tau ,k}}\! - \!{x_{t,i}}} \right) \psi_{t,ki}{\left( {\left(x_{\tau,k}-x_{t,i}\!\right)^2 \!+ \!\left(y_{\tau,k}-y_{t,i}\right)^2 \!+\! h_{\tau,k}^2} \right)^{\frac{\mu_\textrm{NLoS} }{2}-1}}}.
\end{equation} 
Given $\mu_\textrm{NLoS}=2$, (\ref{eq:de}) is simplified to 
$\sum\limits_{i \in \mathcal{U}_{\tau,k}} \sum\limits_{t = 1}^{{F}}2 \left( {{x_{\tau ,k}}\! - \!{x_{t,i}}} \right) \psi_{t,ki}$=0. As a result, $x_{\tau,k}=\frac{{\sum\limits_{i \in \mathcal{U}_{\tau,k}} {\sum\limits_{t = 1}^{{F}} {{x_{t,i}}{\psi _{t,ki}}} } }}{{\sum\limits_{i \in \mathcal{U}_{\tau,k}} {\sum\limits_{t = 1}^{{F}} {{\psi _{t,ki}}} } }}$. Likewise, we can show that $y_{\tau,k}=\frac{{\sum\limits_{i \in \mathcal{U}_{\tau,k}} {\sum\limits_{t = 1}^{{F}} {{y_{t,i}}{\psi _{t,ki}}} } }}{{\sum\limits_{i \in \mathcal{U}_{\tau,k}} {\sum\limits_{t = 1}^{{F}} {{\psi _{t,ki}}} } }}$.

For case b), since $h_{\tau,k}^2 \gg \left(x_{t,i}-x_{\tau,k}\right)^2 + \left(y_{t,i}-y_{\tau,k}\right)^2$, ${{{d_{t,ki}}\left( {\boldsymbol{w}_{\tau,t,k},\boldsymbol{w}_{\tau,t,i}} \right)}} \approx h_{\tau,k}$ and, hence, $\frac{{{h_{\tau ,k}}}}{{{d_{t,ki}}\left( {\boldsymbol{w}_{\tau,t,k},\boldsymbol{w}_{\tau,t,i}} \right)}} \approx 1 \to {\phi _t} = {90^\circ }$. Consequently, $\Pr \left( {{l_{t,ki}^\textrm{LoS}}} \right)=1$. Then, we have $\bar l_{t,ki}\left(\boldsymbol{w}_{\tau,t,k},\boldsymbol{w}_{\tau,t,i}\right) ={l_{t,ki}^\textrm{LoS}}$. The derivation of $\sum\limits_{i \in \mathcal{U}_{\tau,k}} \sum\limits_{t = 1}^{{F}}P_{\tau,t,ki}^{\min}$ will be:
\begin{equation}\small\label{eq:derivation}
\setlength{\abovedisplayskip}{-2 pt}
\setlength{\belowdisplayskip}{ 0pt}
\begin{split}
\frac{{\partial \sum\limits_{i \in \mathcal{U}_{\tau,k}} {\sum\limits_{t = 1}^{{F}} {P_{\tau,t,ki}^{\min }} } }}{{\partial {x_{\tau ,k}}}} &=\sum\limits_{i \in \mathcal{U}_{\tau,k}} {\sum\limits_{t = 1}^{{F}} \mu_\textrm{LoS} \left( {{x_{\tau ,k}}\! - \!{x_{t,i}}} \right) \psi_{t,ki}{\left( {\left(x_{\tau,k}-x_{t,i}\!\right)^2 \!+ \!\left(y_{\tau,k}-y_{t,i}\right)^2 \!+\! h_{\tau,k}^2} \right)^{\frac{\mu_\textrm{LoS} }{2}-1}}},\\
& \approx \sum\limits_{i \in \mathcal{U}_{\tau,k}} {\sum\limits_{t = 1}^{{F}} \mu_\textrm{LoS} \left( {{x_{\tau ,k}}\! - \!{x_{t,i}}} \right) \psi_{t,ki} h_{\tau,k}^{{\mu_\textrm{LoS} }-2}}=0.
\end{split}
\end{equation} 
As a result, $x_{\tau,k}=\frac{{\sum\limits_{i \in \mathcal{U}_{\tau,k}} {\sum\limits_{t = 1}^{{F}} {{x_{t,i}}{\psi _{t,ki}}} } }}{{\sum\limits_{i \in \mathcal{U}_{\tau,k}} {\sum\limits_{t = 1}^{{F}} {{\psi _{t,ki}}} } }}$ and $y_{\tau,k}=\frac{{\sum\limits_{i \in \mathcal{U}_{\tau,k}} {\sum\limits_{t = 1}^{{F}} {{y_{t,i}}{\psi _{t,ki}}} } }}{{\sum\limits_{i \in \mathcal{U}_{\tau,k}} {\sum\limits_{t = 1}^{{F}} {{\psi _{t,ki}}} } }}$.       
This completes the proof.  
\end{proof}
Using Theorem \ref{th:3}, we can find the optimal locations of the UAVs  given the users association and altitude $h_{\tau,k}$ for the two special cases. 
For more generic cases, it is highly challenging to find the optimal UAVs' locations using derivation, since the UAV's altitude depends on the $x$ and $y$ coordinates of the UAV. 
Therefore, we use a learning algorithm given in \cite{21} and \cite{Chen2016Echo} to find a sub-optimal solution. The learning algorithm can learn the network state and exploit different actions to adapt the UAV's location according to the network. After the learning step, each UAV will find a sub-optimal location to service the users in a power efficient way. \vspace{-0.2cm}
\subsection{Implementation and Complexity}
The BBUs implement the conceptor ESN algorithm to predict the content request distribution and mobility patterns for each user. Hence, the BBUs have the prediction information for each user. Moreover, the BBUs have the information of the RRHs cluster and their locations. Therefore, the proposed solutions in Section \ref{section2} and \ref{section3} to the problem in (\ref{eq:sum}) are all implemented in the BBUs. The centralized implementation enables the BBUs to quickly determine the optimal location for each UAV without any information exchange. Furthermore, using a centralized implementation, one can  directly modify the UAVs' locations and users' associations directly when the prediction of the conceptor ESN algorithm is not accurate.

The complexity of the proposed algorithm can be divided into two parts: the conceptor ESN algorithm and the optimization algorithm. The complexity of the conceptor ESN algorithm depends on the number of patterns that needs to be leaned for each user. Therefore, the complexity of the conceptor ESN algorithm will be $O\!\left(U{{ \times 2T} \mathord{\left/
 {\vphantom {{ \times T} H}} \right.
 \kern-\nulldelimiterspace} H}\right)$. Note that, the implementation of the user-RRH association algorithm depends on the number of the RRH clusters, and, hence, the complexity of this algorithm is $O\!\left(E\right)$ where $E$ is the number of the RRH clusters and ${{2T} \mathord{\left/
 {\vphantom {{2T} H}} \right.
 \kern-\nulldelimiterspace} H}$ is the number of the patterns that each user needs to learn. Finally, the algorithm complexity for finding the optimal content to cache and optimal location of each UAV all depend on the number of the UAVs. Thus, the algorithm complexity is $O\!\left(K\right)$ with $K$ being the number of UAVs. 
   
\section{Simulation Results}
For our simulations, the content request data that the ESN uses to train and predict content request distribution is obtained from \emph{Youku} of \emph{China network video index}\footnote{\vspace{-0.27cm}The data is available at \url{http://index.youku.com/}.}. Here, one circular CRAN area with a radius $r = 500$ m is considered with $U=70$ uniformly distributed users and $R=20$ uniformly distributed RRHs.
The detailed parameters are listed in Table  \uppercase\expandafter{\romannumeral5}. Actual pedestrian mobility data is measured from the real data generated at the \emph{Beijing University of Posts and Telecommunications}. For comparison purposes, we investigate: a) optimal algorithm with complete users' information, b) ESN algorithm in \cite{EchoStateNetworks2} to predict the content request distribution and mobility pattern, and c) random caching with ESN algorithm in \cite{EchoStateNetworks2} to predict content request distribution. All statistical results are averaged over 5000 independent runs. The accuracy of ESN prediction is measured by normalized root mean square error (NRMSE) \cite{Jaeger2014Controlling}. 

\begin{table}\footnotesize
  \newcommand{\tabincell}[2]{\begin{tabular}{@{}#1@{}}#2\end{tabular}}
\renewcommand\arraystretch{0.7}
 \caption{
    \vspace*{-0.3em}SYSTEM PARAMETERS}\vspace*{-1.5em}
\centering  
\begin{tabular}{|c|c|c|c|c|c|}
\hline
\textbf{Parameter} & \textbf{Value} & \textbf{Parameter} & \textbf{Value}&\textbf{Parameter} & \textbf{Value} \\
\hline
$F$ & 1000 & $P_R,P_B$ & 20, 30 dBm&$K$&5\\
\hline
$B$ & 1 MHz & $N$ & 25&$C$&1\\
\hline
$\chi _{\sigma_\textrm{LoS}}$ & 5.3 & $H$ & 10 &$T$&120\\
\hline
$ N_{tr}$ &1000 & $d_0$ & 5 m &$f_c$& 38 GHz\\
\hline
$B_v$ & 1 GHz & $\lambda$ & 0.01&$h_{\min}$ & 100 m \\
\hline
$L$ & 1 Mbit & $P_{\max}$ & 20 W&$\delta_{S_i,n}$ & 5 Mbit/s\\
\hline
$\mu_\textrm{LoS}$,$\mu_\textrm{NLoS}$& 2, 2.4 &$N_x, N_s$ &4, 12&$\chi  $ & 15\\
\hline
 $N_w$ & 1000 & $\sigma ^2$ & -95 dBm&$\zeta_1,\zeta_2$&0.5,0.5\\ 
\hline
 $\chi _{\sigma_\textrm{NLoS}}$ & 5.27 & $\beta, \eta $ & 2, 100&$X,Y$& 11.9, 0.13\\ 
\hline
\end{tabular}
 \vspace{-0.5cm}
\end{table}
 
\begin{figure}[!t]
  \begin{center}
   \vspace{0cm}
    \includegraphics[width=16cm]{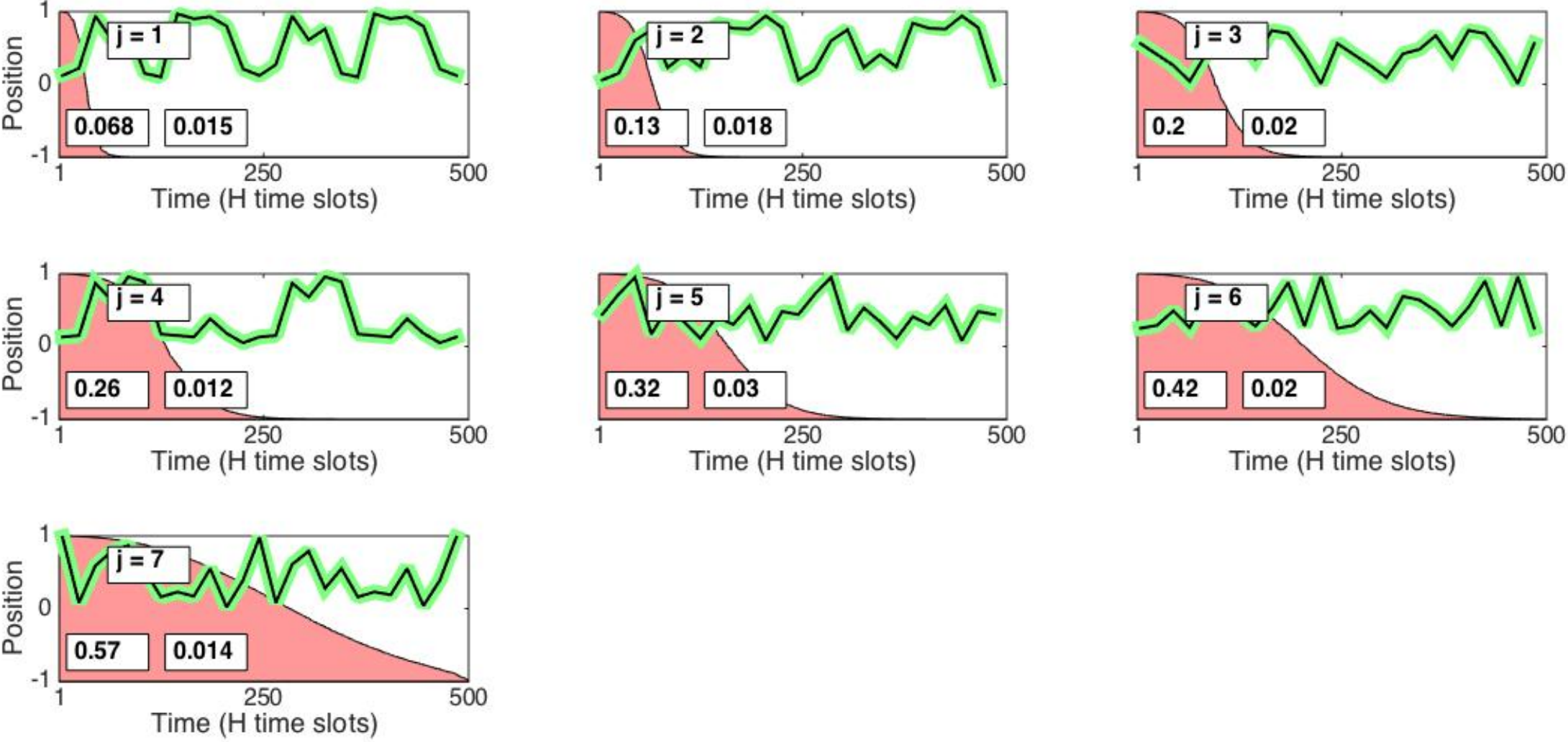}
    \vspace{-0.6cm}
    \caption{\label{figure1} Mobility patterns predictions of Conceptor ESN algorithm. In this figure, the green curve represents the conceptor ESN prediction, the black curve is the real positions, top rectangle $j$ is the index of the mobility pattern learned by ESN, the legend on the bottom left shows the total reservoir memory used by ESN and the legend on the bottom right shows the NRMSE of each mobility pattern prediction.}
  \end{center}\vspace{-1.2cm}
\end{figure}
\vspace{-0cm}

Fig. \ref{figure1} shows how the memory of the conceptor ESN reservoir changes as the number of the mobility patterns that were learned by the conceptor ESN varies. Here, one mobility pattern represents the users' trajectory in one day and the colored region represents the memory used by the conceptor ESN. In Fig. \ref{figure1}, we can see that the memory usage increases as the number of the learned mobility patterns increases. This is due to the fact the conceptor ESN uses a limited memory to learn mobility patterns. From Fig. \ref{figure1}, we can also see that the conceptor ESN uses less memory for learning mobility pattern 2 compared to pattern 6. In fact, mobility pattern 2 is similar to mobility pattern 1, and, hence, the conceptor ESN requires only a small amount of memory to learn mobility pattern 2. However, the conceptor ESN needs to use more memory to learn mobility pattern 6. Clearly, when a new mobility pattern needs to be learned, the proposed approach only needs to learn the difference  between the learned mobility patterns and the new one.

 
 \begin{figure}[!t]
  \begin{center}
   \vspace{0cm}
    \includegraphics[width=9cm]{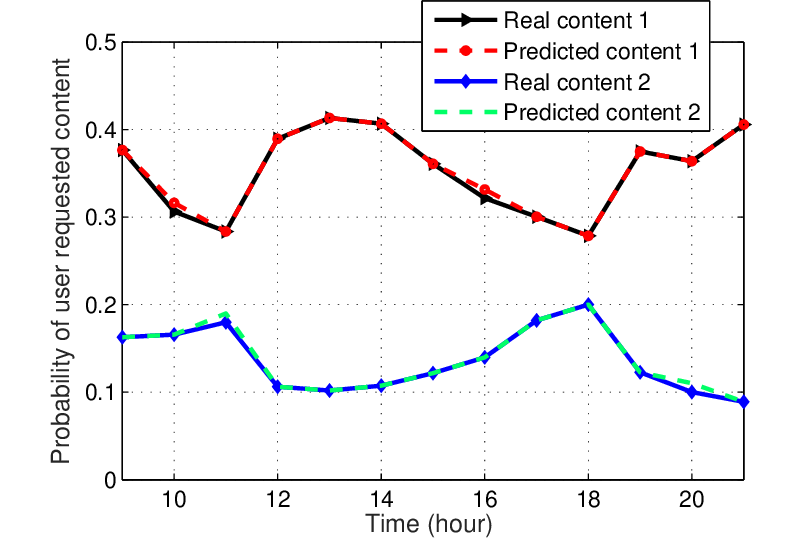}
    \vspace{-0.6cm}
    \caption{\label{figure3}Content request probability predictions.}
  \end{center}\vspace{-1.1cm}
\end{figure}
\vspace{-0cm}  

In Fig. \ref{figure3}, we show the variations of two content request probabilities of a selected user during one day. The user is randomly chosen from the set of users in the network. From Fig. \ref{figure3}, we can see that, the probability with which this user requests content 1 decreases during working hours (9:00-11:00 and 14:00-18:00) and increases at all other times. Similarly, the request probability of content 2 increases during working hours and decreases during the rest of the time. This is due to the fact that content 1 is an entertainment content while content 2 is a work-related content. Fig. \ref{figure3} also shows that the sum of the probability with which this user requests content 1 and content 2 exceeds 0.5 during each hour. This is because the user always requests a small amount of contents during one day.

\begin{figure}
\centering
\vspace{0cm}
\subfigure[]{
\label{figure2a} 
\includegraphics[width=8cm]{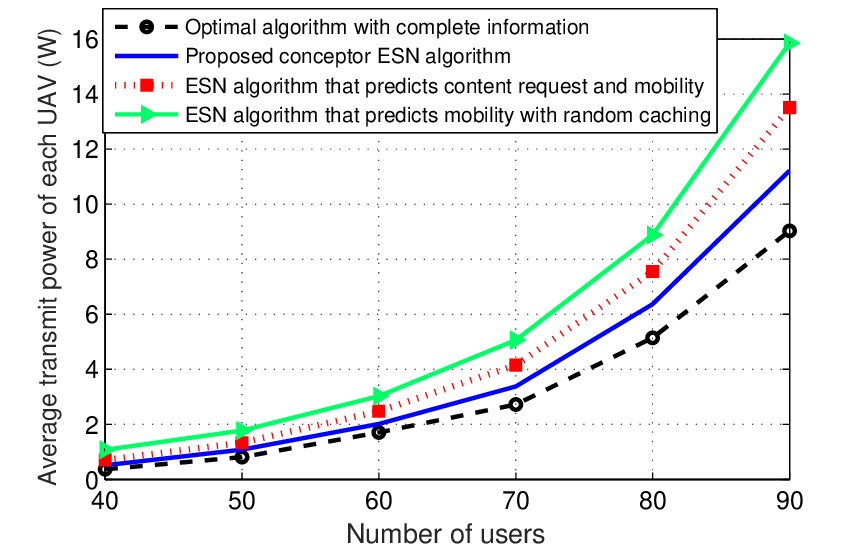}}
\subfigure[]{
\label{figure2b} 
\includegraphics[width=8cm]{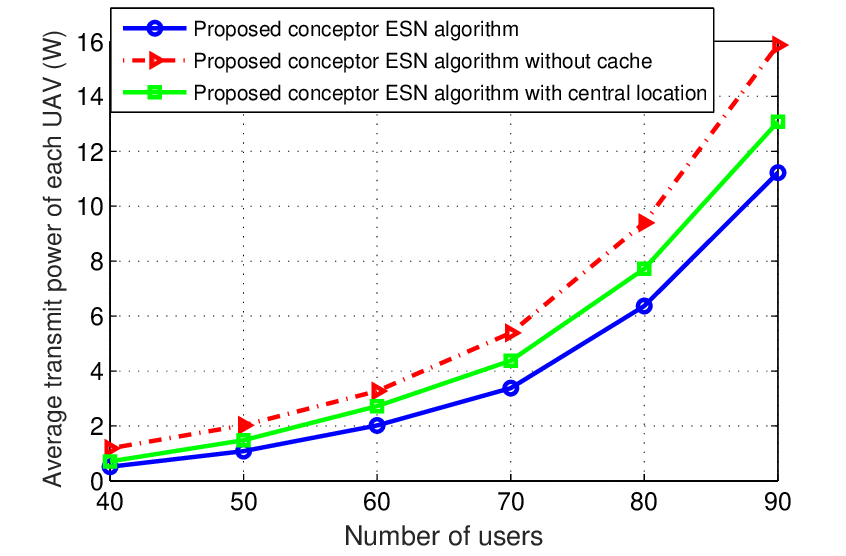}}
  \vspace{-0.8cm}
 \caption{\label{figure2} Average minimum transmit power as the number of the users varies. ($K=5$ and $C=1$.)} 
  \vspace{-0.7cm}
\end{figure}       

Fig. \ref{figure2} shows how the average minimum transmit power of each UAV in a time period changes as the number of the users varies. In Fig. \ref{figure2}, we can see that the average UAV transmit power of all algorithms increases as the number of the users increases. This is due to the fact that the number of the users associated with the RRHs and the capacity of the wireless fronthaul link of UAVs are limited. Therefore, the UAVs need to increase their transmit power to satisfy the QoE requirement of each user. From Fig. \ref{figure2a}, we can also see that the proposed approach can reduce the average transmit power of the UAVs of about 20\% compared to the ESN algorithm used to predict the content request and mobility for a network with 70 users. This is because the conceptor ESN that separates the users' behavior into multiple patterns and uses the conceptor to learn these patterns, can predict the users' behavior more accurately compared to the ESN algorithm. Fig. \ref{figure2b} shows that the proposed algorithm can yield, respectively, 40\% and 25\%
gains with respect to reducing the average transmit power compared to the proposed algorithm without cache
and the proposed algorithm without optimizing the UAVs' locations.

 \begin{figure}[!t]
  \begin{center}
   \vspace{0cm}
    \includegraphics[width=9cm]{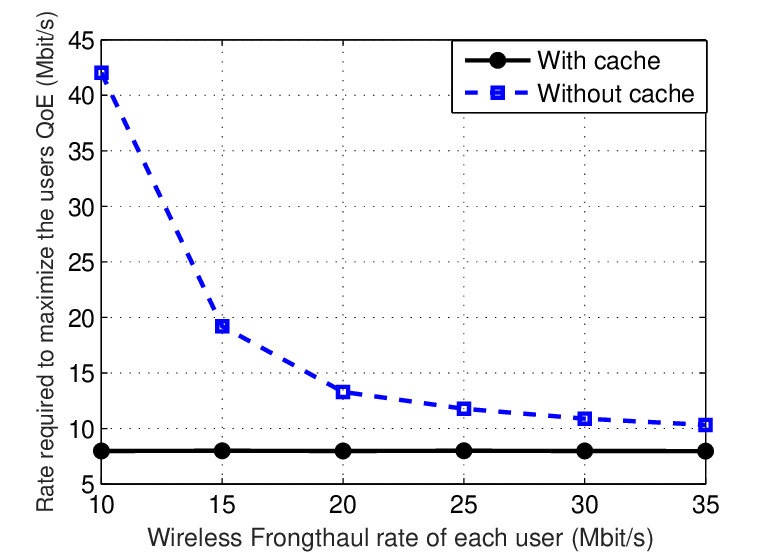}
    \vspace{-0.6cm}
    \caption{\label{figure8} Rate required to maximize the users QoE as the fronnthaul rate of each user changes.}
  \end{center}\vspace{-1.3cm}
\end{figure}
\vspace{-0cm}

Fig. \ref{figure8} shows the rate needed for satisfying the QoE requirement of each user versus the wireless fronthaul rate of each user. In Fig. \ref{figure8}, we can see that the rate required to maximize the users QoE of BBUs-UAV-user link decreases as the wireless fronthaul rate increases. However, the rate needed to maximize the user's QoE of UAV-user link  not change when the fronthaul rate varies. Clearly, the use of caching at the UAVs can significantly reduce the rate required to reach the QoE threshold of each user when the wireless fronthaul rate for each user is low.

\begin{figure}[!t]
  \begin{center}
   \vspace{0cm}
    \includegraphics[width=8.5cm]{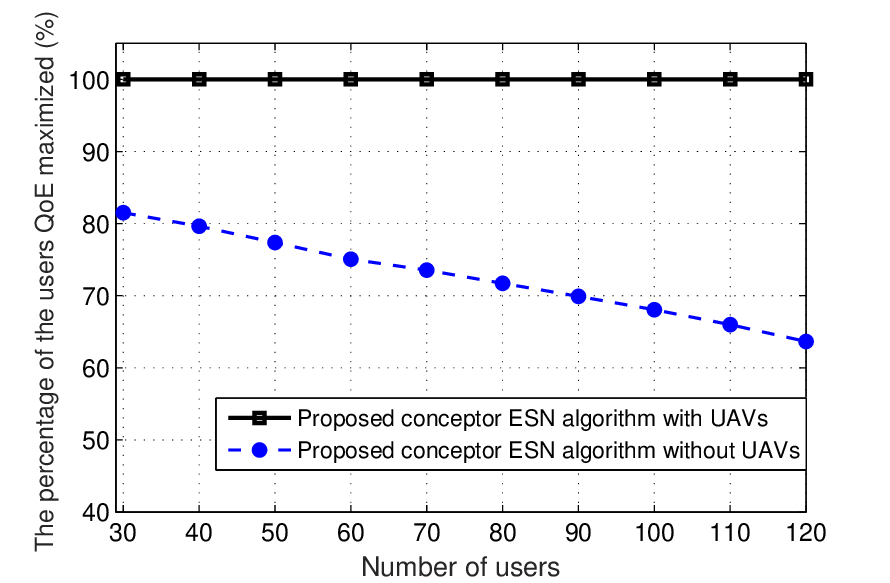}
    \vspace{-0.6cm}
    \caption{\label{figure7}The percentage of the users QoE that is maximized as the number of the users varies}
  \end{center}\vspace{-1.1cm}
\end{figure}
\vspace{-0cm}    

In Fig. \ref{figure7}, we show how the percentage of users with satisfied QoE requirement changes as the number of the users varies. From Fig. \ref{figure7}, we can see that the percentage of the satisfied users decreases as the number of the users increases for the case with no UAVs. However, using the proposed approach, the QoE remains maximum for all number of users when the UAVs are deployed. The proposed algorithm can yield a gain of 61\%
gain in terms of the percentage of the users with satisfied QoE compared to the proposed algorithm without UAVs for the network with 120 users. This is due to the fact that the UAVs can maximize the users' QoE when the RRHs are not able to satisfy the QoE requirements.

 \begin{figure}[!t]
  \begin{center}
   \vspace{0cm}
    \includegraphics[width=9cm]{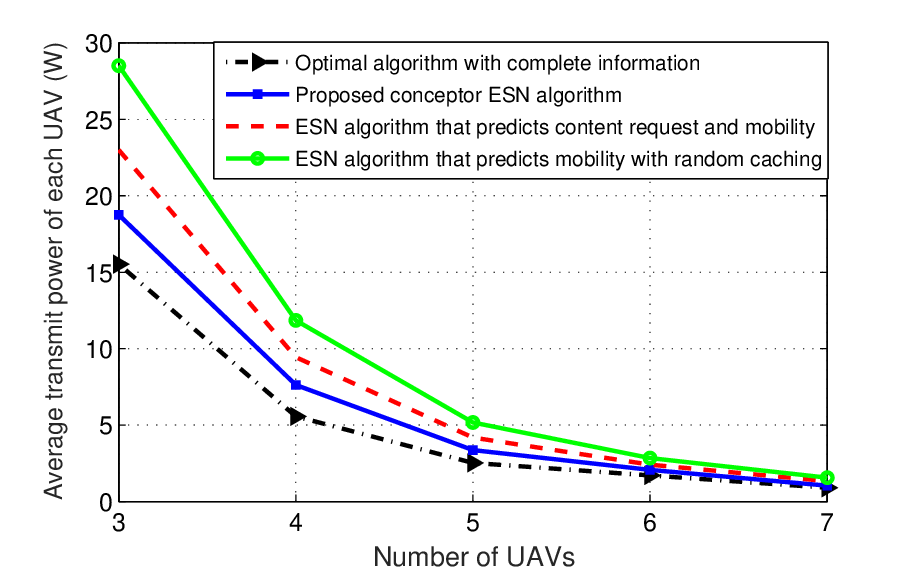}
    \vspace{-0.6cm}
    \caption{\label{figure5}Average minimum transmit power as the number of UAVs changes ($U=70$ and $C=1$).}
  \end{center}\vspace{-1.3cm}
\end{figure}
\vspace{-0cm}

In Fig. \ref{figure5}, we show how the average minimum transmit power of UAVs changes as the number of the UAVs varies. From Fig. \ref{figure5}, we can see that the average minimum transmit power of each UAV decreases from as the number of the UAVs increases. In particular, using the proposed algorithm, the average transmit power of the UAVs decreases by 85\% when the number of UAVs increases from 3 to 7. This is due to the fact that for a higher number of UAVs the number of users associated with each UAV decreases, and, hence, the average transmit power per UAV also decreases.
As shown in Fig. \ref{figure5}, the proposed approach becomes closer to the optimal one as the number of UAVs increases. The reason is that the location prediction error is higher for a lower number of UAVs (or equivalently the clusters).


\begin{figure}[!t]
  \begin{center}
   \vspace{0cm}
    \includegraphics[width=9cm]{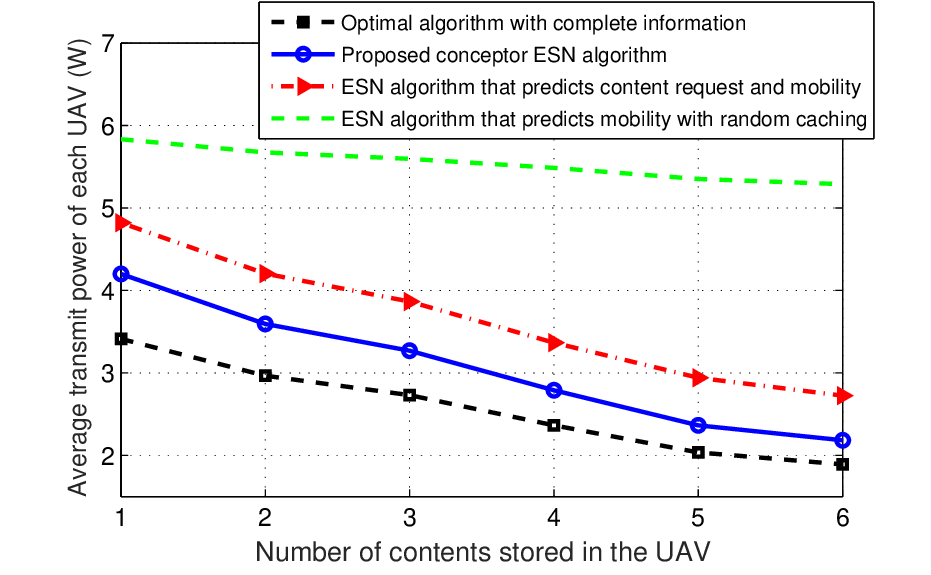}
    \vspace{-0.6cm}
    \caption{\label{figure6}Average minimum transmit power as the number of the contents stored in a UAV cache varies ($U=70$ and $K=5$).}
  \end{center}\vspace{-1.1cm}
\end{figure}
\vspace{-0cm}

Fig. \ref{figure6} shows the average minimum transmit power of each UAV as a function of the number of the contents stored at the UAV cache. As shown in Fig. \ref{figure6}, the average minimum transmit powers of all considered algorithms increase as the number of storage units increases. This is due to the fact that the probability that the requested contents of the users are stored at the UAV cache increases, and, consequently, the UAV will directly transmit the requested contents to the users. Fig. \ref{figure6} also shows that the ESN algorithm that predicts the content request and mobility can yield up to $49\%$ power reduction compared to the ESN algorithm that predicts the mobility with the random caching scheme.


\begin{figure}
\centering
\vspace{0cm}
\subfigure[]{
\label{figure11a} 
\includegraphics[width=8cm]{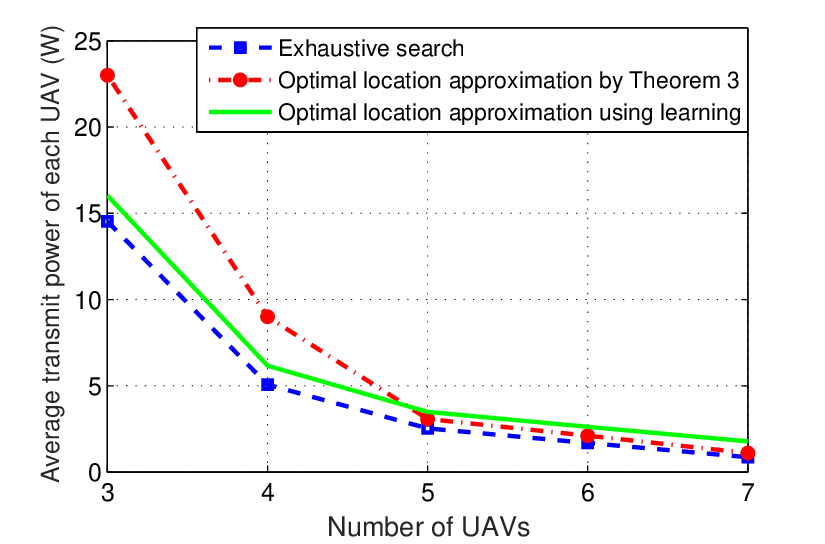}}
\subfigure[]{
\label{figure11b} 
\includegraphics[width=7.7cm]{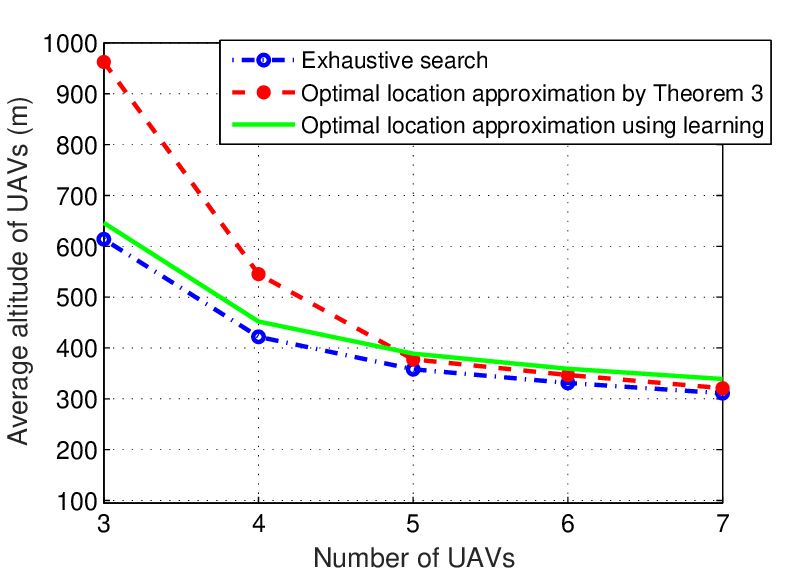}}
  \vspace{-0.2cm}
 \caption{\label{figure11} Average minimum transmit power and average altitude vs. the number of the UAVs.}
  \vspace{-1cm}
\end{figure}     

In Fig. \ref{figure11}, we show how the average transmit power and average altitude of the UAVs change as the number of UAVs varies. In this case, we compare the result of our proposed approach with the optimal result obtained by an exhaustive search method. For the learning algorithm, the interval of the neighboring action of each coordinate is 3 m. Figs. \ref{figure11a} shows that the optimal location of the UAV approached by Theorem \ref{th:3} has only 5.8\% deviation compared to the exhaustive search. Furthermore, as shown in Fig. \ref{figure11b}, by increasing the number of UAVs from 3 to 7, the average altitude of the UAVs decreases from 980 m to 302 m in the proposed algorithm case. This is due to the fact that for a higher number of the UAVs,  each UAV needs to provide coverage for a smaller area and, hence,  it can be deployed at a lower altitude. From Figs. \ref{figure11a} and \ref{figure11b}, we can also see that, as the number of the UAVs increases, the result of Theorem \ref{th:3} approaches the optimal solution that is obtained by the exhaustive search.      
 This is due to the fact that for a higher number of UAVs, the coverage area of each UAV decreases,
and, hence, the approximation condition in Theorem \ref{th:3} will hold with a tighter bound. 
\section{Conclusions}
In this paper, we have proposed a novel framework that uses flying UAVs to provide service for the mobile users in a CRAN system. First, we have presented an optimization problem that seeks to guarantee the QoE rquirement of each user using the minimmum transmit power of the UAVs. Next, to solve this problem, we have developed a novel algorithm based on the echo state networks and concepters. The proposed algorithm allows predicting the content request distribution of each user with limited information on the network state and user context. The proposed algorithm also enables the ESNs separate the users behavior into several patterns and learn these patterns with various non-linear systems. Simulation results have shown that the proposed approach yields significant performance gains in terms of minimum transmit power compared to conventional ESN approaches.\vspace{-0.2cm}    
\bibliographystyle{IEEEbib}
\def\baselinestretch{1.10}
\bibliography{references}
\end{document}